\newif\ifusenix
\usenixfalse          %% ← ACM mode (default)
\makeatletter
\def\input@path{{templates/acmart/}{templates/usenix/}}
\makeatother
\PassOptionsToPackage{table,xcdraw,x11names}{xcolor}

\ifusenix
  \documentclass[letterpaper,twocolumn,10pt]{article}
  \usepackage{templates/usenix/usenix}
\else
  \documentclass[sigplan,twocolumn]{templates/acmart/acmart}
\fi

\usepackage{confswitch}

\usepackage[markup=underlined,commandnameprefix=always]{changes}

\usepackage{letltxmacro}
\makeatletter
\NewDocumentCommand\ulsafe@cite{o o m}{%
  \IfNoValueTF{#1}%
    {\mbox{\ulsafe@origcite{#3}}}%
    {\IfNoValueTF{#2}%
      {\mbox{\ulsafe@origcite[#1]{#3}}}%
      {\mbox{\ulsafe@origcite[#1][#2]{#3}}}}%
}
\newcommand\ulsafe@shield{%
  \LetLtxMacro\ulsafe@origcite\cite
  \LetLtxMacro\cite\ulsafe@cite
}
\setaddedmarkup{{\ulsafe@shield\uline{#1}}}
\setdeletedmarkup{{\ulsafe@shield\sout{#1}}}
\sethighlightmarkup{{\ulsafe@shield\uwave{#1}}}
\makeatother

\acmYear{2026}\copyrightyear{2026}
\setcopyright{cc}
\setcctype[4.0]{by}
\acmConference[SOSP '26]{Symposium on Operating Systems Principles}{September 29--October 2, 2026}{Prague, Czech Republic}
\acmBooktitle{Symposium on Operating Systems Principles (SOSP '26), September 29--October 2, 2026, Prague, Czech Republic}
\acmDOI{10.1145/3830418.3843853}
\acmISBN{979-8-4007-2585-2/26/09}

\ifusenix\else
  \makeatletter
  \patchcmd{\maketitle}%
    {\acmConference@shortname, \acmConference@venue}%
    {\acmConference@shortname, \acmConference@date, \acmConference@venue}%
    {}{\ClassWarning{acmart}{Could not patch copyright block to include date}}
  \makeatother
\fi

\usepackage{tikz}
\usetikzlibrary{positioning, shapes.geometric, arrows.meta, fit, backgrounds, calc, decorations.pathmorphing}
\usepackage{amsmath}

\ifusenix
  \usepackage{amssymb,amsmath,amsfonts,amsthm}
  \usepackage[T1]{fontenc}
  \usepackage[utf8]{inputenc}
  \usepackage{microtype}
  \UseMicrotypeSet[protrusion]{basicmath}
  \usepackage{hyperref}
  \usepackage{graphicx}
  \usepackage{xcolor}
  \usepackage{booktabs}
  \usepackage{etoolbox}
  \usepackage{comment}
  \usepackage{caption}
\fi

\usepackage{listings}

\usepackage{upquote}
\usepackage{grffile}
\usepackage[normalem]{ulem}
\usepackage{multirow}
\usepackage{xspace}
\usepackage{tabularx}
\usepackage{ragged2e}
\usepackage{paralist}

\ifusenix
  \usepackage[american]{babel}
\fi

\ifusenix
  \usepackage{color}
\fi
\usepackage{wrapfig}
\usepackage{balance}
\usepackage{enumitem}
\usepackage{url}
\usepackage{pifont}
\usepackage{mathtools}

\ifusenix
  \usepackage{xkeyval}
\fi

\usepackage{threeparttable}
\usepackage{makecell}

\usepackage[tworuled, vlined]{algorithm2e}
\usepackage{sepfootnotes}

\ifusenix
  \usepackage[compact]{titlesec}
\fi

\usepackage[capitalize]{cleveref}
\crefformat{section}{§#2#1#3}
\usepackage{textcomp}
\usepackage{subcaption}

\usepackage[most]{tcolorbox}
\usepackage{dblfloatfix}

\PassOptionsToPackage{usenames,dvipsnames}{color} % color is loaded by hyperref
\hypersetup{unicode=true,
    colorlinks=true,
    linkcolor=blue,
    citecolor=blue,
    anchorcolor=blue,
    urlcolor=blue,
    breaklinks=true}
\setlist{topsep=1pt, partopsep=0pt, itemsep=1pt, parsep=1pt, leftmargin=10pt}

\makeatletter
\def\maxwidth{\ifdim\Gin@nat@width>\linewidth\linewidth\else\Gin@nat@width\fi}
\def\maxheight{\ifdim\Gin@nat@height>\textheight\textheight\else\Gin@nat@height\fi}
\makeatother
\setkeys{Gin}{width=\maxwidth,height=\maxheight,keepaspectratio}
\makeatletter
\g@addto@macro{\UrlBreaks}{\UrlOrds}
\makeatother

\newcommand\paraspace{\vspace*{0.4ex}}
\providecommand\parab[1]{\paraspace\noindent\textbf{#1}}
\providecommand\parae[1]{\paraspace\noindent\textbf{\textit{#1}}}

\newtoggle{reviewmode}
\newtoggle{anony}

\newcommand{\projurl}[1]{%
    \iftoggle{anony}{%
        URL is hidden for review purpose%
    }{%
        \url{#1}%
    }%
}

\newcommand{\sysname}{\textsc{DDB}\xspace}

\newtheorem{lemma}{Lemma}

\newtheorem{remark}{Remark}
\newcommand{\etc}{\emph{etc.}\xspace}
\newcommand{\ie}{\emph{i.e.,}\xspace}
\newcommand{\eg}{\emph{e.g.,}\xspace}

\newcommand{\secref}[1]{\S\ref{#1}}
\newcommand{\figref}[1]{Figure~\ref{#1}}
\newcommand{\tabref}[1]{Table~\ref{#1}}
\newcommand{\apdxref}[1]{Appendix~\secref{#1}}
\newcommand{\algoref}[1]{Algorithm~\ref{#1}}

\definecolor{teal}{rgb}{0.0, 0.5, 0.5}
\definecolor{olive}{rgb}{0.5, 0.5, 0.0}
\definecolor{pink}{rgb}{1.0, 0.75, 0.8}
\definecolor{lightgray}{gray}{0.75}
\definecolor{mediumgray}{gray}{0.5}
\definecolor{darkgray}{gray}{0.25}
\definecolor{charcoal}{gray}{0.2}
\definecolor{turquoise}{rgb}{0.25, 0.88, 0.82}
\definecolor{coral}{rgb}{1.0, 0.5, 0.31}
\definecolor{navyblue}{rgb}{0.0, 0.0, 0.5}
\definecolor{lime}{rgb}{0.75, 1.0, 0.0}
\definecolor{darkgreen}{rgb}{0.0, 0.5, 0.0}
\definecolor{violet}{rgb}{0.56, 0.0, 1.0}
\definecolor{lightgreen}{rgb}{0.85, 1.0, 0.85}

\definecolor{burgundy}{cmyk}{0.5, 1.0, 0.7, 0.4}
\definecolor{olivegreen}{cmyk}{0.64, 0, 0.95, 0.4}
\definecolor{peach}{cmyk}{0, 0.5, 0.7, 0}
\definecolor{mustard}{cmyk}{0, 0.3, 1, 0}

\newcommand{\circlednum}[2][]{%
    \tikz[baseline=(char.base)]{
        \node[shape=circle,fill=lightgreen,text=black,draw,inner sep=1pt,#1] (char) {\scriptsize #2};
    }%
}

\newcommand{\circledplain}[1]{\tikz[baseline=(char.base)]{
        \node[shape=circle,draw,inner sep=1pt,scale=0.8] (char) {#1};}}

\makeatletter
\renewenvironment{proof}[1][\proofname]{%
    \par\vspace{-0.2\baselineskip}% space above (was -0.5)
    \pushQED{\qed}%
    \normalfont
    \topsep4pt \partopsep0pt % allow a little breathing room
    \trivlist
    \item[\hskip\labelsep\itshape
                #1\@addpunct{.}]\ignorespaces
}{%
    \popQED\endtrivlist\@endpefalse
    \vspace{-0.2\baselineskip}% space below
}
\makeatother

\newcolumntype{N}{>{\raggedleft\arraybackslash}p{1.4cm}} % Right-aligned numbers
\newcolumntype{U}{>{\raggedright\arraybackslash}p{0.8cm}} % Left-aligned units with markers

\ifusenix
    \usepackage{titlesec}
    \titlespacing*{\section}{0pt}{7pt plus 3pt minus 3pt}{3pt plus 3pt minus 2pt}
    \titlespacing*{\subsection}{0pt}{4pt plus 3pt minus 2pt}{1pt plus 3pt minus 1pt}
    \titlespacing*{\subsubsection}{0pt}{4pt plus 3pt minus 2pt}{0pt plus 3pt minus 1pt}
    \titleformat{\section}{\large\bfseries}{\thesection}{1em}{}
    \titleformat{\subsection}{\normalsize\bfseries}{\thesubsection}{1em}{}
\fi

\toggletrue{anony}  % disable author block
\newtoggle{showmarks}

\togglefalse{showmarks} % disable comments and remarks

\iftoggle{showmarks}{
    \newcommand\red[1]{\textcolor{red}{#1}}
    \newcommand\redstrike[1]{\red{\sout{#1}}}
    \newcommand\green[1]{\textcolor{\green}{#1}}
    \newcommand\greenstrike[1]{\green{\sout{#1}}}
    \newcommand\orange[1]{\textcolor{orange}{#1}}
    \newcommand\orangestrike[1]{\orange{\sout{#1}}}
    \newcommand\blue[1]{\textcolor{blue}{#1}}
    \newcommand\bluestrike[1]{\blue{\sout{#1}}}
    \newcommand\purple[1]{\textcolor{purple}{#1}}
    \newcommand\purplestrike[1]{\purple{\sout{#1}}}
    \newcommand\teal[1]{\textcolor{teal}{#1}}
    \newcommand\tealstrike[1]{\teal{\sout{#1}}}
    \newcommand\turquoise[1]{\textcolor{turquoise}{#1}}
    \newcommand\turquoisestrike[1]{\turquoise{\sout{#1}}}
    \newcommand\darkgreen[1]{\textcolor{darkgreen}{#1}}
    \newcommand\darkgreenstrike[1]{\darkgreen{\sout{#1}}}
    \newcommand\lime[1]{\textcolor{lime}{#1}}
    \newcommand\limestrike[1]{\lime{\sout{#1}}}
    \newcommand\olivegreen[1]{\textcolor{olivegreen}{#1}}
    \newcommand\olivegreenstrike[1]{\olivegreen{\sout{#1}}}

    \newcommand{\seojin}[1]{[\orange{\sf\textit{#1 - Seojin}}]}
    \newcommand{\yibo}[1]{[\purple{\sf\textit{#1 - Yibo}}]}
    \newcommand{\junzhou}[1]{[\blue{\sf\textit{#1 - junzhou}}]}
    \newcommand{\draft}[1]{\turquoise{\sf\textit{#1}}}
}{
    \newcommand\red[1]{#1}
    \newcommand\redstrike[1]{\unskip}
    \newcommand\green[1]{#1}
    \newcommand\greenstrike[1]{\unskip}
    \newcommand\orange[1]{#1}
    \newcommand\orangestrike[1]{\unskip}
    \newcommand\blue[1]{#1}
    \newcommand\bluestrike[1]{\unskip}
    \newcommand\purple[1]{\unskip}
    \newcommand\purplestrike[1]{\unskip}
    \newcommand\teal[1]{\unskip}
    \newcommand\tealstrike[1]{\unskip}
    \newcommand\turquoise[1]{\unskip}
    \newcommand\turquoisestrike[1]{\unskip}
    \newcommand\darkgreen[1]{\unskip}
    \newcommand\darkgreenstrike[1]{\unskip}
    \newcommand\lime[1]{\unskip}
    \newcommand\limestrike[1]{\unskip}
    \newcommand\olivegreen[1]{\unskip}
    \newcommand\olivegreenstrike[1]{\unskip}

    \newcommand{\seojin}[1]{}
    \newcommand{\yibo}[1]{}
    \newcommand{\junzhou}[1]{}

    \newcommand{\draft}[1]{}
}

\begin{document}

\title{
  DDB: Source-Level Interactive Debugging for Distributed Applications
}

%% ============================================================
%% STEP 7: Authors (template-specific formatting)
%% ============================================================
\ifusenix
  \author{
    {\rm Paper}\\
    University of Southern California
  }
\else
  % \author{Yibo Yan}
  % \author{Junzhou He}
  % \author{Seo Jin Park}
  % \affiliation{%
  %   \institution{University of Southern California}
  %   \city{}
  %   \country{}
  % }
  % \affiliation{%
  %   \institution{University of Southern California}
  %   \city{}
  %   \country{}
  % }
  % \affiliation{%
  %   \institution{University of Southern California}
  %   \city{}
  %   \country{}
  % }
  % 
  \author{Yibo Yan \quad Junzhou He \quad Seo Jin Park}
  \affiliation{%
    \institution{University of Southern California}
    \city{}
    \country{}
  }
  \renewcommand{\authors}{Yibo Yan, Junzhou He, and Seo Jin Park}
  \renewcommand{\shortauthors}{Yan et al.}
\fi

%% ============================================================
%% STEP 8: Abstract
%%   Explicit inputs let arXiv detect the abstract dependency.
%%   ACM: input before \maketitle.  USENIX: input after.
%% ============================================================
\ifusenix\else
  \begin{abstract}

    Interactive debugging is an effective tool for understanding program behavior at the source level, allowing developers to pause execution, navigate the call stack, and inspect runtime state.
    However, interactive debuggers are designed for single-process execution, and interactive debugging has been widely considered impractical for distributed systems.
    Call stacks stop at process boundaries, debugging state fails to survive infrastructure dynamics, and, most critically, debugger-induced execution pauses trigger catastrophic timeout cascades that destroy the intended debug flow.
    Consequently, developers are forced to abandon live hypothesis testing in favor of unwieldy and iterative log-and-redeploy cycles.

    We present DDB, a source-level interactive debugger that extends interactive debugging capabilities to distributed applications.
    We show that each of these challenges admits a targeted solution.
    To bridge disjoint processes, Distributed Backtrace (DBT) embeds compact causality metadata in every RPC and reconstructs a unified call stack across RPC boundaries.
    To manage the lifecycle of a distributed session, an intent-preserving control plane automatically coordinates and propagates breakpoints across dynamic process sets.
    To make pausing safe, Pause-Erased Time (PET) virtualizes each process's clock, decoupling logical time from physical pauses and preventing timeout cascades.
    DDB integrates with an RPC framework in 10–60 lines of code.
    Evaluated on gRPC, ServiceWeaver, Nu, and Quicksand across up to 122 processes, DDB achieves 30~ms median cross-RPC backtrace latency, sub-5~ms time jump under repeated execution pauses, and adds 1–5\% throughput overhead, comparable to attaching a single-process debugger.
    In a controlled user study, DDB achieves a 100\% fault localization success rate (compared to 38.5\% for baseline tools) with a median localization time of $\sim$8 minutes.

\end{abstract}

\fi

%% ============================================================
%% STEP 9: CCS concepts (ACM-only; ignored by USENIX)
%% ============================================================


% \keywords{}

%% ============================================================
%% STEP 10: \confmaketitle handles ordering correctly
%% ============================================================
\ifusenix\date{}\fi
\confmaketitle
\ifusenix
  \begin{abstract}

    Interactive debugging is an effective tool for understanding program behavior at the source level, allowing developers to pause execution, navigate the call stack, and inspect runtime state.
    However, interactive debuggers are designed for single-process execution, and interactive debugging has been widely considered impractical for distributed systems.
    Call stacks stop at process boundaries, debugging state fails to survive infrastructure dynamics, and, most critically, debugger-induced execution pauses trigger catastrophic timeout cascades that destroy the intended debug flow.
    Consequently, developers are forced to abandon live hypothesis testing in favor of unwieldy and iterative log-and-redeploy cycles.

    We present DDB, a source-level interactive debugger that extends interactive debugging capabilities to distributed applications.
    We show that each of these challenges admits a targeted solution.
    To bridge disjoint processes, Distributed Backtrace (DBT) embeds compact causality metadata in every RPC and reconstructs a unified call stack across RPC boundaries.
    To manage the lifecycle of a distributed session, an intent-preserving control plane automatically coordinates and propagates breakpoints across dynamic process sets.
    To make pausing safe, Pause-Erased Time (PET) virtualizes each process's clock, decoupling logical time from physical pauses and preventing timeout cascades.
    DDB integrates with an RPC framework in 10–60 lines of code.
    Evaluated on gRPC, ServiceWeaver, Nu, and Quicksand across up to 122 processes, DDB achieves 30~ms median cross-RPC backtrace latency, sub-5~ms time jump under repeated execution pauses, and adds 1–5\% throughput overhead, comparable to attaching a single-process debugger.
    In a controlled user study, DDB achieves a 100\% fault localization success rate (compared to 38.5\% for baseline tools) with a median localization time of $\sim$8 minutes.

\end{abstract}

\fi
% \pagestyle{plain}

%% ============================================================
%% Body sections
%% ============================================================

\begin{figure}[t]
      \centering
      \includegraphics[width=\linewidth]{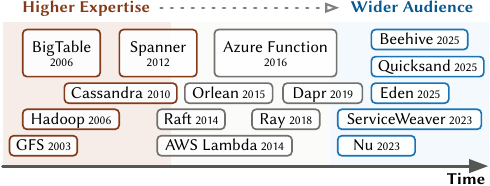}
      \caption{Distributed programming paradigms and infrastructure have become much more accessible to ordinary application developers. The adoption of distributed programming has increased while the required expertise has dropped substantially.}
      \label{fig:ddb-trend}
\end{figure}

\section{Introduction}

Serverless platforms, microservice architectures, and modular programming frameworks increasingly decompose what was once a single-process application into code that executes across multiple services and machines~\cite{jonasCloudProgrammingSimplified2019, sahraei2023xfaas, gannonCloudNativeApplications2017, ghemawatmoderndevelopmentcloud2023, ruanNuAchievingMicrosecondScale2023}.
The same trend spans actor frameworks~\cite{bernsteinorleansdistributedvirtual, moritzraydistributedframework2018, akkabuildrun}, distributed application runtimes~\cite{dapr}, and memory-disaggregated runtimes~\cite{wang2020semeru, ruanaifm, quicksand-nsdi25}.
The result is that application developers---who may have little expertise in distributed systems---are now routinely writing code whose execution spans dozens of processes.
As these paradigms and infrastructure have emerged, the expertise barriers to adopting distributed programming have dropped substantially, as depicted in~\figref{fig:ddb-trend}.
\looseness=-1

For debugging single-process programs, a developer attaches a debugger (\eg GDB~\cite{gdbgnuproject}, LLDB~\cite{lldb}), sets breakpoints, and inspects call stacks, variables, and caller state---a tight hypothesis--inspect--refine loop for understanding \textit{why} code misbehaves.
Once an application becomes distributed, this workflow collapses.
A breakpoint in one process reveals nothing about the remote callers that triggered the current execution: the chain of upstream RPCs, the arguments they passed, and their runtime state remain invisible.
Attaching separate debuggers to every upstream process and locating the correct thread among potentially hundreds~\cite{beschastnikh2016debugging} does not scale well.
Moreover, pausing any process triggers leader re-elections, transaction aborts, or cluster-wide crash recovery, because distributed systems rely on failure-detection timeouts as short as 50--100\,ms~\cite{yuan2012characterizing}.

For traditional distributed infrastructure, such as large-scale storage systems~\cite{ghemawat2003google, corbett2013spanner} and distributed databases~\cite{lakshman2010cassandra}, this absence of interactive debugging has been mitigated by heavy investment in alternative diagnostic tooling.
The teams that build these systems at major organizations invest heavily in structured logging~\cite{yuan2012characterizing}, distributed tracing~\cite{jaeger, OpenTelemetry, openzipkindistributedtracing}, and custom monitoring infrastructure to diagnose faults in production.
Application developers typically lack deep expertise in distributed systems; without interactive debugging, they resort to iterative log-and-redeploy cycles that studies show consume days to weeks---and in extreme cases over a month---per fault~\cite{zhouFaultAnalysisDebugging2021, yuan2014simple}.
As quantified in our user study (\secref{sec:user-study}), participants using baseline tools (GDB, distributed tracing) failed to localize cross-service faults in 61.5\% of cases and often exceeded the 20-minute time limit.

% The tools available today serve fundamentally different goals.
These failures stem from a fundamental mismatch: each existing tool addresses a different aspect of distributed diagnosis, but none provides the interactive workflow that cross-RPC, source-level fault localization requires.
Distributed tracing and logging reconstruct request paths and record event histories~\cite{jaeger, openzipkindistributedtracing, OpenTelemetry}; record-and-replay enables deterministic post-hoc analysis~\cite{rrlightweightrecording, geelsFridayGlobalComprehension2007, vearne2026vearne}.
Interactive debuggers exist for parallel and HPC settings~\cite{acharGoTchaInteractiveDebugger2019, totalviewdebuggerhpc, linaroddt}, but these assume each process can be paused independently without side effects---an assumption that breaks when applications use RPC timeouts for failure detection.
None supports what application developers need: pausing a running distributed execution, navigating the cross-RPC call chain, and examining live runtime state across service boundaries.

We present \sysname, a distributed interactive debugger that provides this workflow.
With DDB, a developer can set a breakpoint once and it is automatically inserted across all relevant replicas, including processes that join later through scaling or restart.
When the breakpoint fires, \sysname defaults to a pause-the-world model (akin to GDB's default behavior), halting all attached processes to provide a consistent global view.
Because all processes are concurrently paused, waiting for manual inspection is safe: \sysname virtualizes each process's view of time so that application-level timeouts and timers are unaffected by the execution pause.
\looseness=-1

We observe that in development and test environments, where coordinated global pauses are feasible, cross-RPC call stacks can be reconstructed by embedding compact causality metadata in every RPC;
dynamic process sets can be managed through an intent-preserving control plane;
and timeout cascades can be eliminated by virtualizing each process's view of time, which requires addressing a temporal anomaly (static application state combined with advancing kernel time) that no existing clock-virtualization mechanism handles~(\secref{sec:model:pet}).
Because interactive debugging naturally targets staging and test environments rather than production, the overhead of these mechanisms is acceptable.

Like traditional interactive debuggers (\eg GDB~\cite{gdbgnuproject}, LLDB~\cite{lldb}), \sysname uses the pause-the-world behavior by default, targeting semantic logic and state errors across RPC boundaries; concurrency bugs that depend on precise thread interleaving fall outside this scope, as they do for all pause-based debuggers.
Under the same assumptions as traditional interactive debuggers, \sysname provides the best debugging experience when debug symbols are embedded in the compiled binaries and the optimization level is tuned down.
\looseness=-1

This paper makes the following contributions:

\begin{itemize}
      \item
            \textbf{Distributed Backtrace (DBT).} Cross-RPC stack reconstruction, transparent to user-level threads, that enables live inspection of caller state across RPC boundaries (\secref{sec:model:dbt}).

      \item
            \textbf{Intent-Preserving Control Plane.} A debug-intent abstraction that propagates breakpoints and commands across replica scaling, node churn, and computation migrations (\secref{sec:model:control-plane}).

      \item
            \textbf{Pause-Erased Time (PET).} Time virtualization via \textit{Virtual Deadline Enforcement} that decouples physical and logical time, enabling safe debugger pauses without timeout cascades (\secref{sec:model:pet}).

      \item
            \textbf{Implementation and Evaluation.} We integrate \sysname with four RPC frameworks in two languages (10--60 LoC each) and evaluate at 122-process scale: 1--5\% throughput overhead, $\sim$30\,ms backtrace latency, and \textless5\,ms time jumps. A controlled user study shows \sysname delivers 100\% cross-service fault localization vs. an aggregated 38.5\% localization rate with baseline tools (\secref{sec:implementation}, \secref{sec:eval}).
\end{itemize}

\section{Background and Motivation}
\label{sec:background}

\subsection{The Developer Workflow and the Missing Capability}
Consider a developer debugging a distributed key-value store where certain read requests return stale values after a write.
The developer can reproduce the issue with a specific sequence of write and read operations issued from a test client.
Distributed tracing identifies the request path: a client request reaches a router, which forwards it to one of several ``storage-node'' replicas.
Structured logging on the storage node confirms that the node served the read from a cached entry and that a cache-invalidation RPC from the router did arrive, but the cached entry was not updated.

At this point the developer knows \textit{where} the stale read occurs and \textit{what} happened, but not \textit{why} the invalidation failed to take effect.
Understanding the root cause requires inspecting runtime state that no existing log captured.
The root cause lies in cross-service runtime state:
\textit{what version identifier did the router embed in the invalidation RPC, what did the storage node's comparator evaluate, and did the invalidation target a different cache slot than intended?}
In single-process development, GDB answers such questions directly: the developer sets a breakpoint, inspects local variables, and walks the call stack.
For a distributed application, no equivalent tool exists.
\looseness=-1

The developer falls back to iterative logging: instrumenting the invalidation handler, redeploying, and re-running the test, with each round revealing only the specific variables the developer chose to log.
After several such cycles spanning two services, the developer discovers that the router was embedding a version field from a stale routing-table entry---a conclusion that a single interactive inspection of the cross-service call chain could have reached in minutes.

Monitoring, tracing, logging, and record-and-replay each address a distinct need, but none provides the capability this workflow demands: live, interactive inspection across service boundaries.
The following subsections detail the entrypoints where interactive debugging begins (\secref{sec:entrypoints}) and the specific challenges that prevent it from working across process boundaries (\secref{sec:challeneges}).

\subsection{Entrypoints for Interactive Debugging}
\label{sec:entrypoints}

An interactive debugging session begins at an entrypoint, such as a manual breakpoint, a watchpoint, an assertion failure (\texttt{SIGABRT}), or a memory fault (\texttt{SIGSEGV}).
In single-process development, these entrypoints cleanly pause execution for live inspection.
In a distributed setting, catching the event is insufficient because the root cause often resides in an upstream remote caller or manifests across dynamic, auto-scaling replicas.
Regardless of the specific trigger, every entrypoint shares a core requirement: once execution pauses, the developer must reconstruct the full cross-process execution context---including the chain of remote callers, their arguments, and their runtime state---to reason about the application.
The next section examines the specific challenges that make this reconstruction difficult.

\subsection{Why Interactive Debugging Breaks Across Process Boundaries}
\label{sec:challeneges}

Distributed execution breaks interactive debugging in three ways (\textcircled{1}--\textcircled{3} in \figref{fig:motivation-barriers}): call stacks end at process boundaries, debugging state does not survive churn, and temporary execution pauses trigger timeout-based failure detection.

\parab{Call stacks stop at the process boundary~(\figref{fig:motivation-barriers}~\textcircled{1}).}
Consider a request that traverses services A → B → C.
The developer suspects C's handler and sets a breakpoint.
When the breakpoint fires, GDB reveals C's local stack but nothing about B or A---neither the arguments B passed nor the execution path that led B to issue the call.
Recovering this context requires attaching a separate debugger to B, identifying the correct thread among potentially hundreds, inspecting B's stack, and repeating for A.
The manual effort scales linearly with call depth.
User-level threading compounds the difficulty: if B uses goroutines, the goroutine that sent the RPC to C may have been descheduled and its OS thread reused, so GDB shows the wrong stack entirely.

\parab{Debugging operations require per-process manual effort~(\figref{fig:motivation-barriers}~\textcircled{2}).}
Consider the 180-process deployment of the \texttt{socialnet} benchmark~\cite{ganOpenSourceBenchmarkSuite2019}.
To debug an issue, a developer must manually identify relevant replicas and insert breakpoints on each individually, which is a cumbersome and tedious process.
Furthermore, the process set changes constantly: auto-scaling adds replicas, rolling restarts replace them, and emerging distributed programming frameworks~\cite{quicksand-nsdi25, quicksand, ruanNuAchievingMicrosecondScale2023, ghemawatmoderndevelopmentcloud2023} introduce runtime computation migrations.
Unless breakpoints automatically follow these state changes and computation migrations, the developer silently loses debugging coverage as execution moves across processes.
\looseness=-1

\parab{Execution pauses cause timeout cascades~(\figref{fig:motivation-barriers}~\textcircled{3}).}
Distributed applications use timeouts extensively to detect and handle failures.
We surveyed timeout thresholds in three open-source distributed systems: LogCabin~\cite{ongaroSearchUnderstandableConsensus} (a Raft implementation), RAMCloud~\cite{ousterhoutRAMCloudStorageSystem2015} (in-memory distributed storage), and Nu~\cite{ruanNuAchievingMicrosecondScale2023} (a distributed programming framework).
Thresholds range from 100 ms (RPC retry) to 500 ms (election timer).
A developer pausing at a breakpoint for even a few seconds triggers every timeout in this range.
The consequences range from harmless (extra heartbeats) to severe (aborted transactions, cluster-wide crash recovery for a healthy node).
The disruption breaks the intended debug flow.
For example, when an execution pause spuriously triggers crash recovery, the system transitions to a recovery state that did not exist before the pause, and the developer ends up inspecting recovery logic instead of the intended code path.

\begin{figure}[t]
    \centering
    \includegraphics[width=\textwidth]{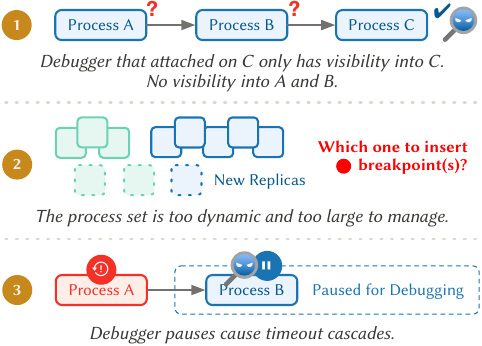}
    \caption{Three challenges of employing interactive debugging in distributed applications.}
    \label{fig:motivation-barriers}
\end{figure}

\begin{figure*}[t]
    \centering
    \includegraphics[width=\linewidth]{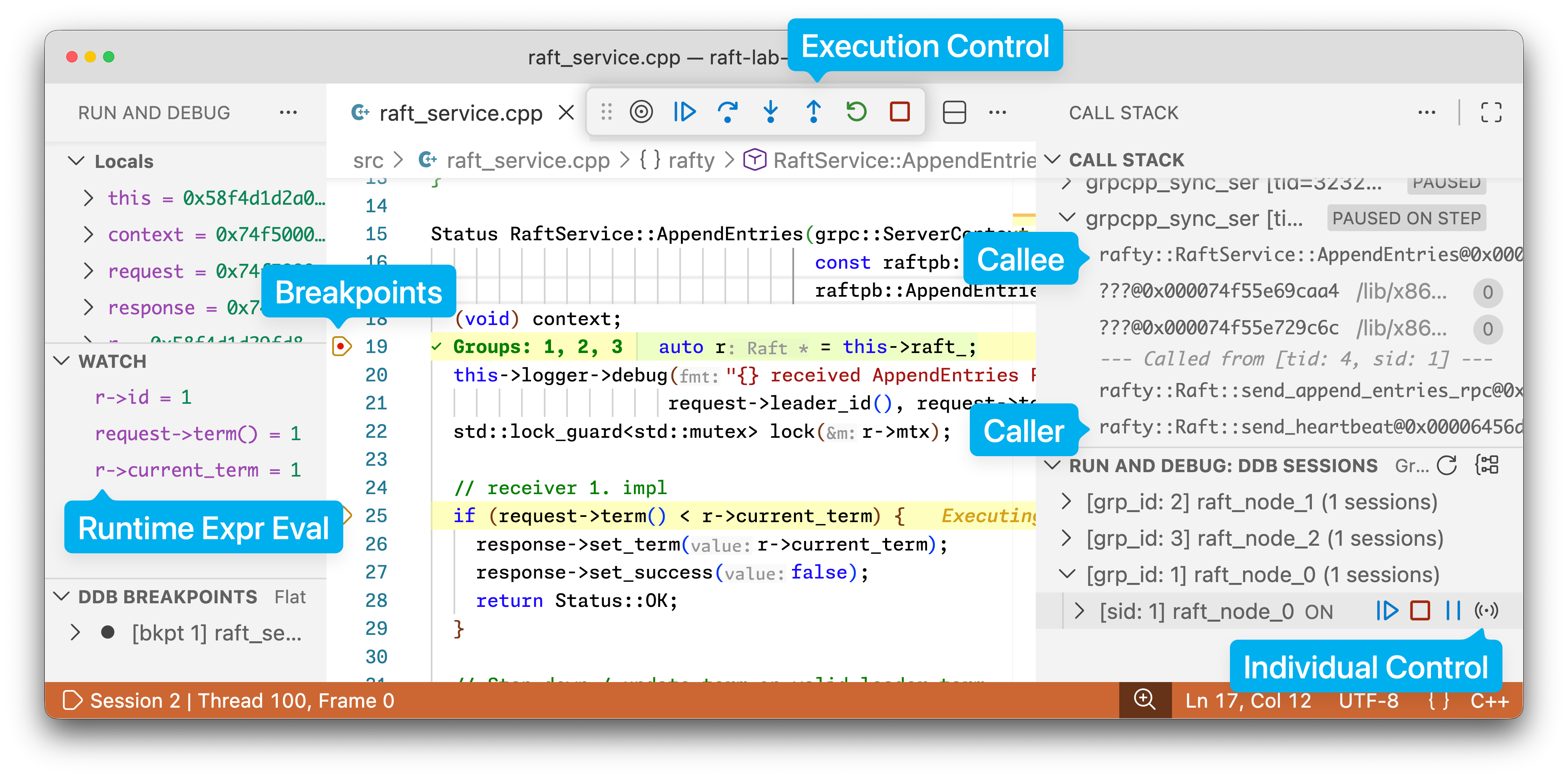}
    \caption{The \sysname VS Code graphical frontend debugging a three-node Raft cluster (with one leader and two followers in the screenshot). Both followers have paused at the local \texttt{AppendEntries} RPC handler (middle). The Call Stack (top-right) displays a Distributed Backtrace extending to the leader's remote calling frame. The \sysname Sessions pane (bottom-right) and debugger control pane (top) allow granular per-process execution control while the rest of the cluster remains globally paused.}
    \label{fig:ddb-showcase}
\end{figure*}

\section{Overview: A \sysname Debugging Session}

To illustrate how \sysname addresses the challenges outlined in \secref{sec:background}, \figref{fig:ddb-showcase} shows a live debugging session on a 3-node Raft cluster.
In Raft, a single \textit{leader} node coordinates the cluster by sending periodic heartbeat RPCs (\texttt{AppendEntries}) to each \textit{follower} replica; if a follower stops receiving heartbeats, it initiates a leader election.
The developer is investigating a bug where followers are improperly rejecting the leader's heartbeats.
Instead of adding iterative logging statements and restarting the cluster, the developer attaches the \sysname VSCode extension and resolves the issue in a single interactive session.
\looseness=-1

\parab{Setting a breakpoint across replicas.}
The developer sets a single breakpoint at the \texttt{AppendEntries} RPC handler.
\sysname's intent-preserving control plane~(\secref{sec:model:control-plane}) automatically applies this breakpoint to all follower replicas in the cluster, without requiring the developer to identify or attach to each process individually.

\parab{Pausing the cluster safely.}
When the leader broadcasts its next heartbeat, both followers hit the breakpoint concurrently and \sysname defaults to a pause-the-world model, halting the entire cluster.
Without \sysname, this pause would be fatal: Raft's election timeouts (typically 100--500\,ms) would expire immediately, triggering a spurious leader election that destroys the state being debugged.
\sysname's Pause-Erased Time (PET)~(\secref{sec:model:pet}) virtualizes each process's clock, keeping the cluster's failure detectors entirely unaware of the execution pause.
\looseness=-1

\parab{Inspecting the cross-RPC call chain.}
With the cluster safely paused, the developer opens the Call Stack pane.
Instead of the local stack frames that a single-process debugger would show, \sysname presents a Distributed Backtrace (DBT)~(\secref{sec:model:dbt}) that extends from the follower's handler back to the leader's \texttt{send\_heartbeat} caller frame.
The developer clicks on the upstream leader frame to inspect the exact arguments sent across the network.
Local variables and expressions are automatically re-evaluated in the context of the selected frame.

\parab{Granular per-process control.}
While the cluster remains globally paused, the developer selects one follower and single-steps it forward (identifiable as ``PAUSED ON STEP'' in the interface), while the other follower stays suspended at the entrypoint.
The developer then evaluates expressions in the Watch pane---comparing the incoming \texttt{request->term()} against the local \texttt{r->current\_term}---and observes how divergent state across replicas affects execution flow, arriving at the root cause.
This coordination is managed by \sysname's intent-preserving control plane~(\secref{sec:model:control-plane}).

This workflow relies on three core technical pillars---Distributed Backtrace, an intent-preserving control plane, and Pause-Erased Time (\figref{fig:ddb-pillar-overview}), as described in the following section.

\begin{figure}[t]
    \centering
    \includegraphics[width=\columnwidth]{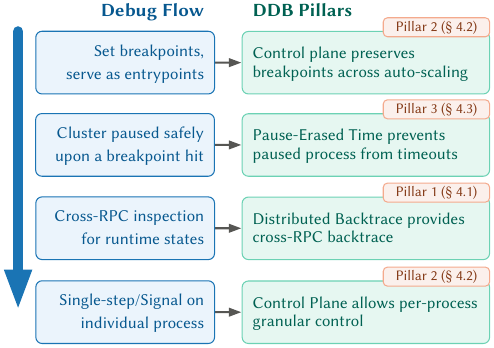}
    \caption{Overview of \sysname's three pillars. Distributed Backtrace (DBT) reconstructs cross-RPC call stacks for live inspection. The intent-preserving control plane automatically propagates debug operations across dynamic process sets and enables granular per-process control. Pause-Erased Time (PET) virtualizes each process's clock to prevent debugger pauses from triggering timeout cascades.}
    \label{fig:ddb-pillar-overview}
\end{figure}

\section{DDB Debugging Model and Design}
\label{sec:debugging_model}

DDB is organized around three pillars: Distributed Backtrace (DBT), an intent-preserving unified control plane, and Pause-Erased Time (PET).
For each pillar, we describe the developer-facing abstraction, the mechanism that realizes it, and the guarantees it provides.

\subsection{Distributed Backtrace (DBT)}
\label{sec:model:dbt}

\parab{Abstraction.}
When execution pauses at any entrypoint, the developer issues \texttt{dbt} command to start a distributed backtrace.
DDB returns a single, unified call stack that begins at the current frame in the paused process and extends backward across each RPC boundary to the root of the request path.
For a request that traversed A → B → C, pausing in C and issuing \texttt{dbt} presents frames from C, then B, then A, in standard GDB backtrace format.
The developer can navigate to any frame (including frames in remote processes) to inspect local variables, read heap-allocated objects, and modify values, much as in single-process debugging. We use this three-service request path as a running example throughout the rest of this section.

\parab{Mechanism.}
Realizing DBT requires solving two problems: locating the correct caller context in a remote process, and presenting a coherent unified call stack.

\textit{Caller-context capture design.}
When process A sends an RPC to B, the distributed backtrace must record enough context at the call site so that the caller's stack can later be reconstructed to form a coherent call stack along the entire causal chain (the exact recursive algorithm is detailed in~\apdxref{app:dss}).
A strawman design records the caller's OS thread ID at the time of sending the RPC and, at reconstruction time, retrieves the thread using that ID to perform stack unwinding.
This fails when the application uses user-level threads (uthreads), such as Go goroutines, CILK workers~\cite{blumofe1995cilk}, Folly fibers~\cite{metafolly}, Shenango~\cite{ousterhout2019shenango}, and Caladan~\cite{fried2020caladana}, which multiplex over a smaller pool of OS threads.
Between the RPC send and the callee's pause, the runtime may reschedule the calling uthread off its original OS thread.
Therefore, unwinding the stack on the caller thread identified by the OS thread ID produces the wrong call frames.

\sysname instead captures the \textit{thread context} (a set of register values that identify the stack) at the call site and embeds it directly in the RPC payload as compact caller-context metadata.
Because the metadata carries the \textit{thread context} rather than a thread identity, it remains valid regardless of thread scheduling decisions made by the runtime.
This embedding is the only change to the RPC protocol and requires approximately 20\;LoC per framework.
At the callee process, the RPC handler extracts this metadata and stores it in a lightweight injected extraction frame (\texttt{DDB::Backtrace::extraction}, as shown in~\figref{list:meta-embed}), making it available on the stack for \sysname to read at reconstruction time.

\textit{Cross-process stack reconstruction.}
When \texttt{dbt} is issued at process C, \sysname locates the injected extraction frame in C's stack and reads the caller-context metadata, which identifies B by the embedded IP address and PID.
Crucially, the reconstruction of B's stack does not happen inside C's address space.
Instead, \sysname's centralized control plane reads the metadata, contacts the debugger agent attached to process B out-of-band, and fetches the necessary stack memory segments.
The agent then temporarily restores B's thread context and performs standard DWARF stack unwinding using B's debug symbols.
\sysname repeats this procedure for B's caller, iterating until no further caller-context metadata exists.
It then concatenates all stack frames into a single stack and presents it to the developer.
The entire procedure is transparent and the developer sees one contiguous call stack.
\looseness=-1

% \textit{User-level thread compatibility.} 
% Many distributed applications often use lightweight threads, such as Go's goroutine, CILK, Folly fibers, Caladan, \etc \yibo{cite}, which multiplex user-level threads (uthreads) over a smaller pool of OS threads.
% Between the moment a uthread issues an RPC and the moment the callee pauses, the runtime may reschedule that uthread off its original OS thread.
% % When a uthread issues an RPC, its OS thread is typically reused by the runtime.
% GDB, which reasons stacks by OS thread, would then display the wrong stack entirely.
% % If the developer later pauses the callee, the caller's OS thread may be running a different uthread; GDB's view of that thread's stack would be wrong.
% % DDB addresses this by embedding the thread context (several register values) at the call site, not the OS thread ID.
% DDB avoids this problem by embedding the thread context (several register values) at the call site rather than na\"{\i}vely recording the caller OS thread id.
% When the developer issues \texttt{dbt} from the callee process, DDB restores the saved thread context in caller before unwinding, reconstructing the correct uthread stack regardless of how the runtime has rescheduled OS threads in the interim.

\parab{Concurrent request handling.}
Distributed backtrace correctly reconstructs call stacks for concurrent requests as long as the corresponding caller contexts remain valid.
For example, suppose a client uses three threads to issue three concurrent requests to a server.
If each sending thread remains blocked on its RPC (\ie synchronous blocking RPCs), \sysname can independently perform \texttt{dbt} for any of the three request chains by placing a breakpoint in the server-side handler logic.
Because each RPC carries the caller's thread context, \sysname can uniquely identify and reconstruct the corresponding caller stack from each server-side handler thread.
Consequently, concurrent requests remain isolated during backtrace reconstruction rather than being conflated with one another.
However, distributed backtrace does not currently support event-driven concurrency, where the caller stack may be unwound or reused immediately after the request is issued.
We discuss this limitation in more detail below.
\looseness=-1

\parab{Guarantees.}
DBT provides three guarantees.
\textit{Completeness:} the cross-RPC caller chain is reconstructed to the root or to the earliest reachable caller, correctly recovering each caller's stack regardless of user-level thread rescheduling.
\textit{Inspectability:} every reconstructed frame exposes locals, arguments, and heap-accessible objects for reading and modification.
\textit{Order preservation:} frames appear in reverse of request-path order (from callee to ancestor callers).

\parab{Edge cases.} Under replication, DBT reconstructs the chain for the specific request that reached the paused process.
If a caller has terminated before the callee pauses, DBT reports a truncated backtrace with a ``caller terminated'' annotation.
If an upstream framework lacks \sysname integration, DBT stops at that boundary and reports the truncation reason.

\begin{figure}[!t]
\begin{lstlisting}[
    language=C++,
    frame=lines,
    basicstyle=\small\fontfamily{zi4}\selectfont,
    tabsize=1,
    numbersep=5pt,
    framesep=5pt,
    columns=fullflexible,
    keepspaces=true,
    showstringspaces=false,
    commentstyle=\color{teal},
    literate={0}{{\textcolor{blue}{0}}}1
             {1}{{\textcolor{blue}{1}}}1
             {2}{{\textcolor{blue}{2}}}1
             {3}{{\textcolor{blue}{3}}}1
             {4}{{\textcolor{blue}{4}}}1
]
0  RpcFlow::Parse 
1  |   RpcMethod::RunHandler 
2  |   |   DDB::Backtrace::extraction  // injected
2  |   |   |   meta = extractor();     // local var
3  |   |   |   RPCServer::user_handler
4  |   |   |   |   // User defined RPC handlers
\end{lstlisting}
\vspace{-2pt}
\caption{Illustration of the frame injection. Blue numbers on the left indicate the frame number, where 0 is the outermost frame and 4 is the innermost frame. Frame~\#2 is the injected frame. The local variable \texttt{meta} is inside the \texttt{extraction} frame.}
\label{list:meta-embed}
\end{figure}

\parab{Limitations.}
Distributed backtrace performs stack unwinding under the assumption that the relevant caller stacks remain available.
It therefore requires the caller's stack frames to remain valid when the backtrace is collected and currently supports only synchronous, blocking RPCs; in particular, it is incompatible with event-driven, asynchronous, or stackless concurrency models that do not preserve a physical caller stack across an RPC boundary.
When \sysname encounters such an asynchronous RPC boundary, \texttt{dbt} returns a truncated call stack because the current design does not propagate caller thread context across that asynchronous RPC boundary.
Accordingly, reconstructing a complete distributed call chain requires every RPC along the chain to preserve a valid ancestor caller context, as is the case for synchronous, blocking RPCs.
\looseness=-1

This is a longstanding challenge for interactive debuggers rather than a limitation specific to \sysname's design: any debugger that offers conventional backtrace capability, where the developer can switch to an arbitrary ancestor frame and inspect its values, must assume that caller stacks remain valid.
To capture causality across asynchronous requests, prior work either falls back to tracing and logging, giving up interactive inspection~\cite{ravindranath2012appinsight, ma2018kernelsupported, weng2021argus}, or offers only a causal view of ancestor callers by function name, without inspectable on-stack values~\cite{stanleyCausewayMessageorientedDistributed}.

\subsection{Intent-Preserving Unified Control Plane}
\label{sec:model:control-plane}
\parab{Abstraction.}
We define a \textit{debug-intent} as a triple: a source-code \textit{location} (file and line number), a \textit{scope} identifying which processes to target (``all replicas of service X'' or ``a specific process''), and a debugger \textit{command} (``set breakpoint'', ``continue'', ``evaluate expression'', ``send signal'', \etc).
The control plane resolves each intent against the current application topology and propagates it to every matching process.
This unified control plane lowers the cognitive overhead of operating individual debuggers on distributed applications.
Crucially, to provide a consistent view of distributed state, the control plane enforces a \textit{pause-the-world} policy by default: whenever any attached process pauses (due to a breakpoint hit, single-step operation, or manual interrupt), the control plane automatically broadcasts a pause command to all other processes attached to the debugging session.

For example, if a developer sets a breakpoint targeting ``all replicas of service B'', \sysname will preserve the debug-intent and install the breakpoint on every current replica and on every future replica that joins, restarts, or receives a migrated computation.
The developer never manages a process list.
This abstraction reduces the operational complexity of debugging from per-process to per-intent, regardless of deployment size.
\looseness=-1

\parae{Pause-the-world Policy.}
Pausing the world prevents in-flight requests from advancing while the developer inspects the system, decoupling the human speed of inspection from the execution speed of the cluster.
For example, distributed backtrace could pause only the processes involved in a call chain, but other processes could continue executing and make the inspected state stale or inconsistent with the rest of the system.
This challenge is inherent to interactive debugging: GDB similarly stops all threads by default, while providing a non-stop mode for independently controlling individual threads.
Likewise, \sysname allows developers to disable pause-the-world and pause only the processes involved in a debugging action, at the cost of accounting for state changes caused by processes that remain active.

\parae{Comparison with Consistent Snapshots.}
Classic snapshot algorithms~\cite{chandy1985distributed} \textit{record} a causally consistent cut while the system keeps running, but the recorded state may not match any global state the execution actually passed through, and it is a static copy that the developer cannot probe or modify.
Pause-the-world instead \textit{freezes} the execution itself: since a paused process cannot send messages, the resulting cut is always consistent, and once the last process stops, the frozen state is the actual global state of the system at that instant.
In-flight messages stay buffered in the network stack and are delivered on resume, so the developer inspects, and may modify, a real reachable state rather than a recorded copy.
The cost is that the pause is not instantaneous: later-paused processes keep executing briefly during the skew window (\secref{sec:model:pet} discusses the resulting timeouts).

\parab{Logical Groups and Scope Resolution.}
\sysname organizes processes into \textit{logical groups}.
By default, processes running the same binary belong to the same group.
In the case of microservices, the default binary-based grouping policy puts all replicas of the same service under the same group.
In the \texttt{socialnet} example with 36 services at 5 replicas, the developer sees 36 groups rather than 180 processes; expanding a group reveals its members.

When a developer specifies a source-code location, \sysname performs \textit{automatic scope resolution}: it identifies which logical groups map that source file, presents them as selectable scopes, and lets the developer choose a target.
This eliminates the need to know which binary hosts a source file---a non-trivial question in distributed application deployments.
Internally, DDB maintains a two-level \textit{source-map}, which maps source-code paths to logical groups and maps logical groups to running processes.

% \begin{figure}[t]
%     \centering
%     \includegraphics[width=0.95\columnwidth]{figures/source_map_new.png}
%     \caption{The two-level source-map: source-code paths map to logical groups, and logical groups map to processes.}
%     \label{fig:source-map}
% \end{figure}

% ~(\figref{fig:source-map}).

\parab{Intent Preservation Under Dynamics.}
Distributed applications are dynamic: auto-scaling, rolling restarts, and failover continuously change the process set.
The control plane tracks these changes and continuously enforces each intent on the correct set of processes.
Specifically, it maintains a continuous invariant: an intent is active on a process if and only if the process matches the intent's logical scope and its loaded binary maps the specified source-code location.
When a new process joins the system, the control plane evaluates this invariant and applies the relevant intents immediately upon debugger attachment.
Conversely, when a process terminates, its debugging state is cleanly garbage-collected.
The intent is a \textit{specification} (\eg ``all replicas of service X should have a breakpoint inserted at line 42'') and \sysname continuously enforces it.

For frameworks that support computation migration (Nu~\cite{ruanNuAchievingMicrosecondScale2023}, Quicksand~\cite{quicksand, quicksand-nsdi25}, ServiceWeaver~\cite{ghemawatmoderndevelopmentcloud2023}), breakpoints follow the computation automatically: because the control plane targets logical scopes rather than physical processes, the receiving process simply evaluates the invariant and applies the breakpoint locally.
Heap state requires additional handling, as a migrated caller's heap may no longer reside on the original process; 
\sysname can be extended to support these frameworks by supplying custom logic that locates and restores heap segments whenever the execution is paused~(\apdxref{app:heap_restore}).

We note that \sysname preserves intents across process churn, but it does not preserve a process's runtime state across restarts.
Once a process terminates, its heap and stack state is lost, and \sysname cannot restore it on the restarted process.
For example, if a callee handler holds the caller context of a process that has since restarted, a distributed backtrace from that handler is truncated because the original caller process no longer exists.

\parab{Guarantees.}
\textit{Invariant enforcement:} the control plane maintains the invariant that an intent is active on a process if and only if the process matches the intent's scope.
When an intent is created, the command completes only after every current in-scope process acknowledges installation.
When a process joins, it is held at an attach-time barrier until all pending in-scope intents are installed, so a matching process never runs without its intents.
\textit{Minimal developer expression:} developers operate at source-file and logical-group granularity, abstracted from physical process churn or state migration.
\textit{Composability:} intents compose orthogonally with \sysname's other primitives; a developer can specify an intent, allow the system to continuously enforce it across churn, and subsequently issue commands like \texttt{dbt} against the paused state without managing the underlying topology.

\subsection{Pause-Erased Time (PET)}
\label{sec:model:pet}

\parab{Problem.}
Every debugger-induced pause---breakpoint hit, single-step, manual interrupt---advances the system clock while the application is stopped.
From the application's perspective, time jumps by the pause duration.
Distributed applications fundamentally rely on these strict physical timing invariants for failure detection: a missed heartbeat triggers a leader election~\cite{ongaro2014search}, a late RPC triggers a retry storm, and an expired lease triggers client disconnection.
A developer pausing a process for ten seconds triggers every timeout below that threshold, forcing the distributed system into recovery procedures that permanently alter the state and destroy the intended debug flow.
Because of this coupling, pause-based interactive debuggers have historically been considered impractical in distributed contexts~\cite{beschastnikh2016debugging, acharGoTchaInteractiveDebugger2019}.

\begin{table}[]
\centering
\resizebox{\columnwidth}{!}{%
\begin{tabular}{@{}lll@{}}
\toprule
\textbf{Semantic} & \textbf{Description}            & \textbf{Example POSIX APIs}                                \\ \midrule
\texttt{get\_time}    & Get the current timestamp         & \begin{tabular}[c]{@{}l@{}}\texttt{clock\_gettime} \\ \texttt{gettimeofday} \end{tabular}   \\ \midrule
\texttt{sleep\_until} & Sleep until the absolute deadline & \begin{tabular}[c]{@{}l@{}}\texttt{pthread\_cond\_timedwait}\\ \texttt{sem\_timedwait}\end{tabular} \\ \midrule
\texttt{sleep}             & Sleep for the relative duration & \begin{tabular}[c]{@{}l@{}}\texttt{sleep}\\ \texttt{nanosleep}\end{tabular} \\ \bottomrule
\end{tabular}%
}
\caption{The Pause-Erased Time (PET) layer provides three semantics; each covers a subset of POSIX time APIs.}
\label{tab:pet-semantic}
\end{table}

\parab{Abstraction.}
\underline{P}ause-\underline{E}rased \underline{T}ime (PET) gives each debugged process a pause-erased view of time: a clock in which all debugger pauses are invisible.
Timestamps and timer waits operate against this clock.
From the application's perspective, it was never paused, thus mitigating the timeout cascades caused by debugger-induced execution pauses.
PET covers three broad categories of time semantics that implement a large portion of timeout and timer logic observed in distributed systems: reading the current timestamp (\eg \texttt{get\_time}), sleeping until an absolute deadline (\eg \texttt{sleep\_until(deadline)}), and sleeping for a relative duration (\eg \texttt{sleep(duration)}).
These semantic categories map to a wide range of underlying POSIX APIs, as demonstrated in \tabref{tab:pet-semantic}.

\parab{Pause-offset Accumulation.}
\sysname records a timestamp when a given debuggee process is paused and another when it is resumed.
The difference is the measured pause duration.
\sysname accumulates these into a running cumulative pause offset.
Before resuming execution, \sysname publishes this offset to a local shim layer, which sits between the user-level process and the \texttt{libc} library to intercept time API calls.
% (user-level process $\rightarrow$ interception layer $\rightarrow$  .
This shim layer intercepts POSIX time API calls and subtracts the cumulative offset from the \texttt{libc} and kernel's return value, erasing all pauses from the process's view of time.

\parab{Virtual Deadline Enforcement.}
The offset mechanism handles \texttt{get\_time} correctly, but a subtler problem arises for processes that are \textit{sleeping} when the debugger pauses them.
On resume, the kernel observes that the real-time deadline has passed and wakes the thread immediately---before the pause-adjusted deadline has been reached.
This \textit{premature wakeup} can produce spurious timeouts.

PET prevents this with the \textit{Virtual Deadline Enforcement} mechanism.
The shim layer tracks the pause-adjusted absolute deadline for every active timer wait.
On every return from the underlying \texttt{libc} sleep primitive---whether a genuine wakeup, a signal interruption, or a spurious wakeup from a debugger resume---the shim layer reads the current PET.
If the current PET is earlier than the adjusted deadline, the deadline has not been reached in pause-erased time; the shim layer re-arms the wait for the remaining duration and returns to the kernel sleep.
Application code does not execute until the PET deadline is genuinely reached.
The shim layer re-traps the thread whenever a resume would wake it prematurely.

Virtual Deadline Enforcement is a fundamental requirement for interactive distributed debugging.
Without it, dynamically erasing pauses would successfully adjust \texttt{get\_time} but fail to prevent timeout cascades triggered by in-flight sleeps. Unlike virtual clocks in discrete event simulation~\cite{jefferson1985virtual, fujimoto1990parallel} or chaos engineering~\cite{hommelWolfcwLibfaketime2024, kingsbury2025jepsenio, barr2014tardis, king2005debugging}---which use static offsets or state rollbacks---PET handles a fundamentally different temporal anomaly: the application state remains static while real time progresses, requiring dynamic offset growth upon every pause.

\parab{Guarantees.}
PET satisfies two invariants; the formal correctness proofs are detailed in~\apdxref{app:pet-prove}.
\begin{itemize}
    \item
          \textit{Monotonicity.} PET never decreases.
          Two consecutive \texttt{get\_time} calls return non-decreasing values regardless of intervening pauses.

    \item
          \textit{Timer correctness.} Time-blocking operations (both absolute deadlines via \texttt{sleep\_until} and relative durations via \texttt{sleep}) return only when PET reaches or exceeds their target.
          Any pause during the sleep, and any premature kernel wakeup from a debugger resume, causes the shim to re-arm.
          Application post-sleep code executes only when the intended PET duration has genuinely passed.
\end{itemize}

These invariants ensure that application logic that is correct with respect to real time remains correct with respect to PET.
Timeouts fire when and only when the application intends, as if no debugger were attached.
PET works for both monotonic and wall-clock time APIs~(\apdxref{app:wall-clock}).

\parab{PET Effectiveness under Global Pause.}
We distinguish between three types of timeouts related to \sysname-attached processes: local timeouts, cross-process timeouts, and external timeouts.
Local timeouts depend entirely on local state (\eg a periodic timer fired every $N$ seconds) without relying on states from other components or services.
PET is unconditionally effective for local timeouts because it directly virtualizes the local clock to minimize execution-pause-induced time jumps.
In contrast, cross-process timeouts depend on communication with other attached processes, \eg a timer triggered when a process receives no response from a peer before a deadline, heartbeats between a leader and its followers, or a cache server waiting on a database reply.
For these cross-process timeouts, PET is effective only if the global pause is tightly synchronized: a pending timeout can fire spuriously only if the pause-time skew between the involved processes exceeds the timeout's remaining slack.
Therefore, to properly mitigate cross-process timeouts with PET, the communication delay between \sysname and attached processes must be bounded so that all attached processes are paused nearly simultaneously.
We state this as a deployment assumption rather than a property \sysname enforces: the control plane is expected to share a low-latency network with the attached processes, as is typical in the development and testing environments \sysname targets, so the pause skew remains well below common timeout values (\eg heartbeat intervals of hundreds of milliseconds).
Finally, external timeouts originate from uncaptured services outside the attached cluster; we discuss them in the limitations below.

\parab{Limitations.}
PET maintains virtual time purely within the boundaries of the attached cluster.
However, PET cannot virtualize the physical flow of time for external, true-time observers.
If the cluster interacts with external systems that rely on physical time---such as an uninstrumented cloud storage service enforcing a strict lease expiration, or a third-party API gateway---those external services will perceive the developer's pause as a genuine timeout.
Thus, \sysname's temporal guarantees are strictly intra-cluster: developers must either mock time-sensitive external dependencies or attach them to \sysname.
For example, when debugging client--server interactions, attaching the load-generating clients as well allows \sysname to pause the world and mask time jumps on both sides, so developers can debug the interaction without triggering unintended timeouts between the clients and the server.
\looseness=-1
\section{Implementation}
\label{sec:implementation}

We implement \sysname on the \texttt{x86\_64} architecture and Linux. We also add experimental support for \texttt{aarch64}.

\begin{figure}
    \centering
    \includegraphics[width=0.95\linewidth]{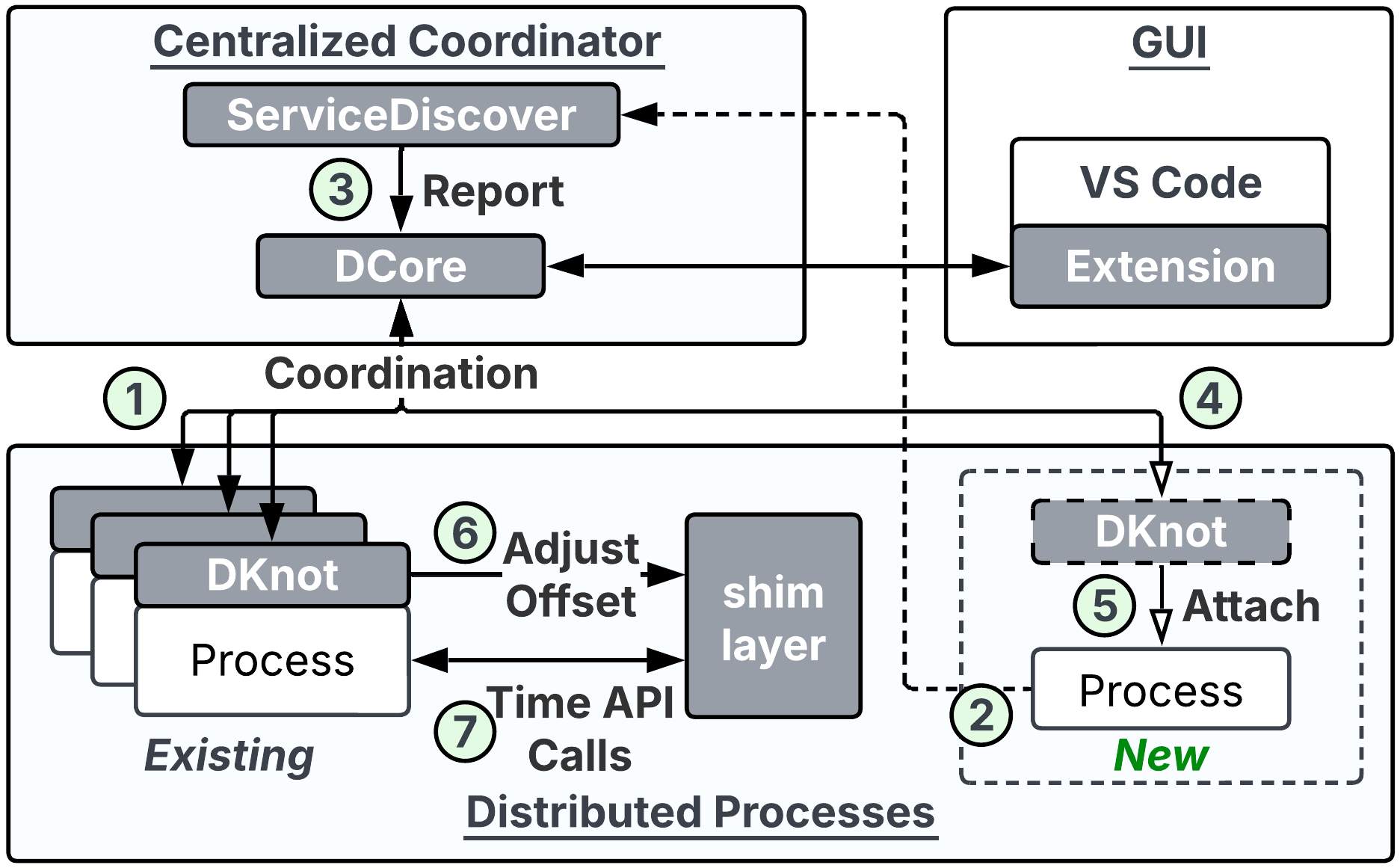}
    \caption{The design overview of \sysname. Components in gray are \sysname components.}
    \label{fig:ddb-overview}
\end{figure}

\parab{System Architecture.} \figref{fig:ddb-overview} illustrates \sysname's architecture, which comprises a centralized coordinator, a graphical frontend, and per-process distributed components.
The \textit{VS Code Extension} serves as the graphical user interface (GUI) for the developer.
This frontend communicates with the \textit{Centralized Coordinator}, which uses \textit{ServiceDiscover} to track dynamic process topologies and report them to \textit{DCore}.
\textit{DCore} acts as the central engine, orchestrating debugging sessions and distributing coordination messages to the relevant \textit{Distributed Processes}.
At the application layer, \sysname pairs each target process with a \textit{DKnot} agent (an extended version of GDB).
\textit{DKnot} attaches to the underlying process, relays debug-intents to it, and implements the \texttt{dbt}.

To prevent timeout cascades, \textit{DKnot} sends offset adjustments to a locally interposed shim layer.
This shim layer intercepts POSIX time API calls from the process to enforce Pause-Erased Time. 
The shim layer approach uses \texttt{LD\_PRELOAD} to interpose on top of the \texttt{libc} interfaces and works with vDSO transparently (\apdxref{app:vdso}).
Finally, a lightweight \textit{DConnector} integration library (not pictured) embeds caller-context metadata within RPC payloads.

\parab{Interaction Between Components.}
The lifecycle of a debugging session follows the numbered steps in \figref{fig:ddb-overview}.
During normal execution, DCore maintains continuous coordination (\circlednum{1}) with all active DKnot agents.
When a new process spins up (\circlednum{2}), ServiceDiscover detects its presence and reports (\circlednum{3}) to DCore updating the global topology.
DCore then immediately initializes a new DKnot instance (\circlednum{4}) which performs an attach (\circlednum{5}) to the new process, ensuring seamless debug-intent propagation across dynamic scaling.
When a process hits a breakpoint, DCore pauses the world and the DKnot agents calculate the precise timestamp of the pause time.
Upon resuming, each DKnot calculates the cumulative pause duration and publishes it as an offset to its local shim layer (\circlednum{6}).
The shim layer then intercepts all subsequent time API calls from the application (\circlednum{7}), seamlessly applying the offset to enforce the Pause-Erased Time abstraction.

% \usepackage{booktabs}
% \usepackage{multirow}
% \usepackage{graphicx}
\begin{table}[]
    \centering
    % \resizebox{\columnwidth}{!}{%
    % \scriptsize
    % \footnotesize
    \small
    \begin{tabular}{@{}lrrc@{}}
        \toprule
                                              & \textbf{Component} & \textbf{LoC}               & \textbf{Language}    \\ \midrule
        \multirow[t]{4}{*}{DDB}               & DCore              & 14,867                     & Rust                 \\
                                              & DKnot              & 1,401                      & Python               \\
                                              & DConnector         & 1,954                      & C/C++/Go             \\
                                              & Adapter            & $\approx$700               & TypeScript           \\ \midrule
        \multirow[t]{3}{*}{Framework Support} & gRPC               & $\approx$20                & C++                  \\
                                              & Nu                 & $\approx$30                & C++                  \\
                                              & Quicksand          & $\approx$60                & C++                  \\
                                              & ServiceWeaver      & $\approx$10                & Go                   \\ \midrule
        \multirow[t]{2}{*}{Ported Apps}
        % & raft          & 1,735                      & C++                  \\
                                              & socialnet          & 3,154                      & Go                   \\ \midrule
        % Total                        &               & \multicolumn{1}{l}{17,207} & \multicolumn{1}{l}{} \\ \midrule
        Total                                 &                    & \multicolumn{1}{l}{22,196} & \multicolumn{1}{l}{} \\ \midrule
    \end{tabular}%
    % }
    \caption{LoC for DDB's components, excluding \texttt{libfaketime}. Calculated by \texttt{cloc} \cite{aldanialaldanialcloc2025}, excluding blanks and comments.}
    \label{tab:codebase}
\end{table}

\parab{Codebase.}
To evaluate our realization, we added DDB support for four frameworks, gRPC \cite{GRPC}, Nu \cite{ruanNuAchievingMicrosecondScale2023}, Quicksand \cite{quicksand-nsdi25}, and ServiceWeaver \cite{ghemawatmoderndevelopmentcloud2023}.
We summarize the LoC in \tabref{tab:codebase}.
To facilitate the evaluation of ServiceWeaver-based applications, we ported \texttt{socialnet} from DeathStarBench \cite{ganOpenSourceBenchmarkSuite2019} to ServiceWeaver.

The modest integration effort reflects \sysname's design philosophy: the communication framework layer is the ideal interposition point.
Integrating \sysname's caller-context embedding within the framework's payload handling path guarantees transparent cross-process stack reconstruction without heavy application-level instrumentation.
This one-time, minimal integration cost (under 60 LoC per framework) instantly unlocks interactive distributed debugging for all downstream application code.
We defer the discussion of more implementation-specific aspects to~\apdxref{app:impl-nits}.

\subsection{Limitations}
\parab{Programming Language Support.} Our prototype currently assumes GDB as the underlying debugger. Therefore, our prototype does not work for binaries and programming languages that GDB cannot interpret and manage.

\parab{PET.} In our realization, DDB only provides PET on top of POSIX time APIs. Thus, it does not work in applications that directly use raw \texttt{rdtsc} or NIC hardware timestamp. However, three API semantics and two invariants ensure our techniques are transferable to other cases.
PET's core mechanism is independent of the time source; 
only the interposition point differs. 
For example, as \texttt{rdtsc} cannot be intercepted via \texttt{LD\_PRELOAD}, supporting it requires developers to route rdtsc through a PET wrapper. 
The wrapper can be implemented as a real rdtsc instruction offset by a global cumulative counter. 
We measured the worst-case time for the first wrapper call after the execution resume to be 70 ns, and successive calls have no measurable overhead (16 ns with or without wrapper). 

\parab{In-flight Packets During Execution Pauses.}
\sysname relies on the host OS's TCP receive buffers to queue in-flight packets during a pause, avoiding complex infrastructure-level checkpointing~\cite{burtsev2009transparent}.
Thanks to PET, the application processes these buffered packets upon resumption exactly as if they arrived over a congested link.
Crucially, under \sysname's pause-the-world behavior, clients and upstream services are also paused, halting transmission.
Thus, TCP buffers trivially absorb the in-flight traffic without exhausting capacity, provided the cluster is not exposed to unmanaged, continuous external traffic.

\section{Evaluation}
\label{sec:eval}

We evaluate \sysname to answer the following key questions:

\begin{enumerate}
    \item How does \sysname's integration effort compare to existing tools? (\secref{sec:user-study})
    \item Does \sysname improve diagnostic reliability and efficiency over existing tools? (\secref{sec:user-study})
    \item Is \sysname responsive enough for interactive debugging at distributed scale? (\secref{sec:responsiveness})
    \item Does PET effectively prevent timeout cascades? (\secref{sec:time-drifts})
    \item What is the runtime overhead of the DDB distributed control plane? (\secref{sec:ddb_overhead})
\end{enumerate}

\parab{Setup.} We evaluate DDB responsiveness on a cluster of 24 physical servers hosted on Chameleon~\cite{keahey2020lessons}. The servers are \texttt{compute\_cascadelake\_r} instances located at the CHI@TACC site. Each server has an Intel Xeon Gold 6240R (24 cores, 2.4GHz), 187GB RAM, and Broadcom BCM57414 NIC. We evaluate DDB's PET effectiveness in minimizing time gaps on our local server, which has Intel Xeon Gold 5420+ (28 cores, 2.00 GHz) and 256GB RAM. We measure the DDB metadata embedding and attachment overhead on a cluster of 12 physical servers hosted on CloudLab~\cite{Duplyakin+:ATC19}. These servers are \texttt{c6525-25g} instances (16-core AMD 7302P at 3.00GHz, 128GB RAM, Mellanox ConnectX-5 NIC).
\looseness=-1

\parab{Evaluation Scope.}
Because our prototype is currently tightly coupled with GDB, Java-based systems are outside the scope of our evaluation.
Supporting Java-based systems would require \sysname to integrate a Java-compatible debugger, such as \texttt{jdb}, as its backend.
We also found relatively few open-source distributed systems that are implemented in C or C++ and use gRPC as their primary inter-component communication framework, further limiting the range of suitable evaluation targets.

\parab{Applications.}
Due to the evaluation scope constraint, we evaluate \sysname with applications and frameworks that are readily available.
To measure responsiveness and overhead across our target research frameworks (ServiceWeaver, Nu, Quicksand), we use the \texttt{socialnet} microservice benchmark as a standardized baseline.
For gRPC, we evaluate a C++ Raft consensus implementation, which is derived from a course lab assignment and is representative of internal cluster communication.
Finally, we isolate PET effectiveness using a synthetic application (\secref{sec:time-drifts}) paired with a survey of real-world timeout thresholds.

\subsection{User Study}
\label{sec:user-study}

\begin{figure}
    \centering
    \includegraphics[width=\linewidth]{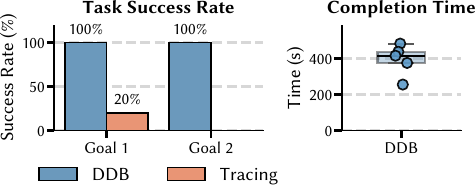}
    \caption{\textbf{Left:} Integration success rates for \sysname vs. OpenTelemetry, where success requires viewing linear (G1) and concurrent (G2) caller--callee traces. \textbf{Right:} Completion time for successful \sysname integrations (both G1 and G2). OpenTelemetry is not shown as no participant finished the full integration within the 20-minute time limit.
    }
    \label{fig:user-study-integration}
\end{figure}

To evaluate \sysname's ease of adoption and its impact on developer velocity, we conducted a two-part controlled user study.
The study quantifies the gap between \sysname and the status quo of debugging tools for distributed applications across two dimensions: the effort required to integrate the tools, and the diagnostic efficacy of \sysname during real debugging tasks.
Prior to the main study, we piloted our methodology by recruiting one volunteer for each study to conduct a trial-run, allowing us to refine the tasks and improve the clarity of the study materials.

\parab{Study 1: Integration Effort.}
The first study evaluates the operational overhead of deploying distributed debugging capabilities.
We recruited 5 participants for this study, comprising graduate researchers from the ECE and CS departments with distributed systems experience.
Each participant was asked to integrate \sysname and OpenTelemetry into a C++ codebase using official
documentation.
Each tool integration task was allocated 20 minutes.
A successful integration required two goals: tracing execution backward from callee to
caller with local state visible (G1), and observing a concurrent fan-out pattern (G2).
\sysname's integration inherently satisfies both goals, whereas OpenTelemetry requires
distinct, context-specific instrumentation for each.
As shown in \figref{fig:user-study-integration}, all 5 participants completed both goals
with \sysname (median 400 seconds), while only 1 completed G1 with OpenTelemetry and
none completed G2.

\begin{figure}
    \centering
    \includegraphics[width=\linewidth]{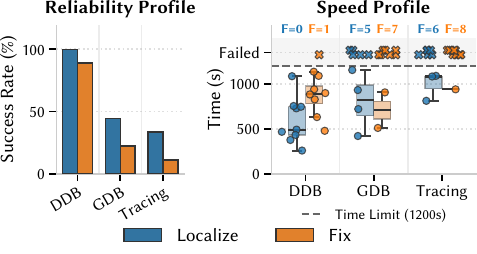}
    \caption{Aggregate reliability and diagnostic speed profiles across all trials.
        \sysname reliably enabled participants to localize the root cause in 100\% of trials, whereas baseline tools failed in the majority of cases and frequently reached the 1200-second timeout limit.
        Note that baseline boxplots may appear compact or low in certain cases; this reflects survivorship bias, as these distributions include only the few participants who localized the fault, while the majority timed out and are placed as failed cases.
    }
    \vspace{0.5em}
    \label{fig:debug-cases-1}
\end{figure}

\parab{Study 2: Diagnostic Efficacy.}
The second study evaluated diagnostic efficiency and reliability.
This study employed a within-subjects (repeated measures) design to compare \sysname against GDB and OpenTelemetry (OTel).
A total of 9 participants were recruited for this study; each participant was required to solve three separate bug cases, utilizing a different tool for each case.
To mitigate potential order effects, learning effects, and fatigue, the presentation sequence of the tools was partially counterbalanced using a $3 \times 3$ Latin Square design: \sysname-GDB-OTel, GDB-OTel-\sysname, or OTel-\sysname-GDB.
% While this fractional design does not completely eliminate asymmetrical carryover effects, it sufficiently controls for primary order and fatigue factors given the study's constraints.
Across all conditions, participants had full access to structured logging (\texttt{spdlog}~\cite{melman2026gabime}) and a traffic-replay utility (\texttt{grpcreplay}~\cite{vearne2026vearne}), so the OTel baseline effectively represents the modern status quo: using a distributed tracer to identify the failing node, followed by iterative logging to find the root cause.
For each task, we recorded the success rate and the total time required to localize and fix the bug, capped at a 1200-second (20-minute) time limit.

\parae{Methodology.}
Comparing \sysname with GDB demonstrates the usability benefits of a unified control plane and evaluates the effectiveness of PET in masking unintended timeouts that would otherwise disrupt the debugging workflow.
Comparing \sysname with OTel evaluates \sysname's debugging efficiency by contrasting interactive, source-level debugging with an iterative logging-based approach.

\parae{Case Selection Criteria.}
The three test cases were selected to highlight different diagnostic challenges:
\begin{inparaenum}[(1)]
    \item
    A \textit{segmentation fault} where an out-of-bounds memory access is non-deterministically triggered across two out of three Raft processes.
    \item
    A \textit{logic error} within the Raft consensus implementation, where a stale leader fails to correctly update its term after recovering from a network partition.
    \item
    A \textit{distributed deadlock} triggered by concurrent leader elections that paralyze cluster operations.
\end{inparaenum}
All three cases are organic, drawn from course teaching staff implementations.
All three require reasoning across multiple Raft processes (three to five nodes) and involve timing-sensitive protocol behavior, making them representative of cross-service faults.
All three study cases are reproducible. In two cases (Case 1 and 2), however, the specific node on which the failure manifests is nondeterministic. \sysname addresses this uncertainty through all-replica breakpoints and pause-the-world execution, allowing the failure to be captured regardless of which replica exhibits it.
\looseness=-1

\figref{fig:debug-cases-1} summarizes the reliability and speed profiles across all trials.
\sysname enabled participants to localize the root cause in 100\% of trials and fix the bug in 89\% within the time limit.
In contrast, GDB achieved only 44\% localization and 22\% fix rates; OTel achieved 33\% and 11\%, respectively.
For successful trials, \sysname also demonstrated high efficiency.
Participants using \sysname attained a median localization time of approximately 500 seconds across all bug cases.
The case-by-case breakdown (\figref{fig:debug-cases-2}) confirms this consistency: in Cases 1 and 3, baseline tools timed out for most participants, while in Case 2, participants using baseline tools required substantially more time than those using \sysname.
Overall, these results indicate that \sysname accelerates distributed debugging and makes complex diagnostic tasks easier to reason about.

\begin{figure}
    \centering
    \includegraphics[width=\linewidth]{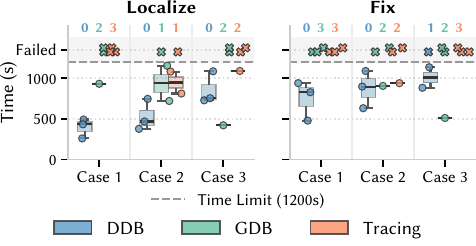}
    \caption{Case-by-case comparison of time-to-localize and time-to-fix across different bug cases.
        Case 1: Random Segmentation Fault; Case 2: Logic Error; Case 3: Distributed Deadlock.
        The GDB outlier in Case 3 corresponds to a participant who resolved the fault based on prior intuition rather than systematic tool-guided diagnosis.
    }
    % For complex cross-process scenarios (Cases 1 and 3), baseline tools experienced severe timeouts.
    % Even in less severe scenarios (Case 2), \sysname consistently accelerated diagnosis workflows.}
    \label{fig:debug-cases-2}
\end{figure}

\parae{Qualitative Analysis.}
Analysis of screen recordings of participant debugging procedures revealed two recurring patterns behind baseline failures.
Participants assigned to OTel spent a disproportionate share of their session performing the instrumentation and switching between the trace visualizer and log collector, rather than performing diagnostic reasoning.
This instrumentation and navigation overhead consumed substantial time and, in many cases, participants exceeded the time limit.
GDB was effective for single-process faults but scaled poorly: participants managed three processes, but at five could no longer efficiently attach to and navigate the relevant threads.
In timing-sensitive cases (Cases 2 and 3), pausing a process for inspection triggered RPC timeouts from its callers, which invalidated the cluster state and forced participants to restart the debugging session from scratch.

\subsection{Responsiveness}
\label{sec:responsiveness}

\begin{table}[t]
\centering
\resizebox{\columnwidth}{!}{%
% \footnotesize
\small
\begin{tabular}{@{}rrrrrrr@{}}
\toprule
\multirow{2}{*}{\textbf{Command}} & \multirow{2}{*}{\textbf{\begin{tabular}[c]{@{}r@{}}Scale\\ (\# pods)\end{tabular}}} & \multicolumn{5}{c}{\textbf{Latency (ms)}} \\ \cmidrule(l){3-7} 
                          &     & \textbf{Count} & \textbf{Mean} & \textbf{Median} & \textbf{P95} & \textbf{P99} \\ \midrule
\multirow{3}{*}{\texttt{dbt}}      & 38  & 2,348          & 32.5          & 14.2            & 128.5        & 153.9        \\
                          & 62  & 4,303          & 32.1          & 13.9            & 144.5        & 156.3        \\
                          & 122 & 14,573         & 26.5          & 13.3            & 110.1        & 160.1        \\ \midrule
\multirow{3}{*}{\texttt{continue}} & 38  & 37             & 2.2           & 2               & 3            & 4.2          \\
                          & 62  & 8              & 2.8           & 2.7             & 3.6          & 3.6          \\
                          & 122 & 7              & 4.7           & 4.8             & 4.9          & 4.9          \\ \bottomrule
\end{tabular}%
}
\caption{Measured latencies for two commonly used commands, \texttt{dbt} and \texttt{continue}, under different cluster sizes. The \textbf{count} column shows the number of times the command is handled by DDB.}
\label{table:66_indi_cmds}
\end{table}

An interactive debugger must respond to user commands and reconstruct state fast enough to maintain the developer's train of thought at human timescales.
\tabref{table:66_indi_cmds} presents the end-to-end latency of two basic control plane operations---executing a distributed backtrace (\texttt{dbt}) and continuing execution---across varying cluster sizes (38, 62, and 122 processes).
Due to the space constraint, we present the full measurement of command handling latencies in~\apdxref{app:ext-cmd-handling}.

\sysname exhibits stable, interactive-grade latency across all tested scales.
The \texttt{continue} command consistently finishes under 5~ms because debug-intent distribution is heavily parallelized.
Crucially, evaluating \texttt{dbt} latency serves as a rigorous stress test of \sysname's responsiveness under heavy concurrent load.
When execution pauses globally, the VS Code frontend automatically issues a burst of \texttt{dbt} commands---one for every thread across every paused process---resulting in over 14.6\,k concurrent requests at the 122-pod scale.
Despite this sheer volume of requests triggering a costly, iterative stack-stitching algorithm, \texttt{dbt} latency remains under 200~ms even at the tail.

\parab{Call Depth.} Because distributed backtrace is an iterative algorithm, \texttt{dbt} handling latency is primarily bounded by call depth.
A longer call chain entails larger latency, as \sysname must cross more RPC boundaries, query more DKnot agents, and perform iterative stack restorations and frame stitching.
To evaluate this impact, \tabref{table:bt_remote_latency} isolates the end-to-end latency of a single \texttt{dbt} command as a function of RPC call depth.
Reconstructing the distributed stack takes a median of 18.2~ms
at a depth of two hops, and 36.1~ms at a depth of four hops. A
recent study~\cite{huye2023lifting} reports that typical call depths observed in
modern cloud applications average around 3 to 4. At these
depths, the sub-50\,ms \texttt{dbt} latency remains well within the
target threshold for a fluid interactive experience.
For long call chains, we discuss a mitigation in~\apdxref{sec:dbt-latency} that harnesses parallelism for a lower latency.

\subsection{PET Effectiveness}
\label{sec:time-drifts}

\label{sec:time-jump}

To evaluate the effectiveness of the PET mechanism, we answer two questions: how much overhead does the interception layer introduce, and how effectively does it virtualize time to prevent unintended timeouts?

% \begin{table}[]
% \small
% \centering
% \begin{tabular}{@{}crrrr@{}}
% \toprule
% \textbf{Hop Count} & \textbf{Mean} & \textbf{P50}  & \textbf{P95} & \textbf{P99}  \\
% \midrule
% 1 hops     & 116.1 & 104.3  & 169.9 & 212.4  \\
% 2 hops     & 178.6 & 158.6 & 269.1 & 340.1 \\
% 3 hops    & 233.1 & 215.8  & 311.6 & 372.8 \\
% \midrule
% \end{tabular}
% \caption{Distributed backtrace command latency (ms) with a different number of hops in calling chain.}
% \label{tab:backtrace-perf}
% \end{table}

\begin{table}[t]
    \centering
    % \footnotesize
    \small
    \begin{tabular}{@{}rrrrrr@{}}
        \toprule
            \multirow{2}{*}{\textbf{Call Depth}} & \multicolumn{5}{c}{\textbf{Latency (ms)}} \\ \cmidrule(l){2-6}
             & \textbf{Mean} & \textbf{Median} & \textbf{StdDev} & \textbf{P95} & \textbf{P99} \\ \midrule
            2  & 18.3   & 18.2   & 0.2 & 18.7  & 18.8   \\
            3  & 29.1   & 28.4   & 2.7 & 29.6  & 33.1   \\
            4  & 36.7   & 36.1   & 2.6 & 37.9  & 38.6   \\
            5  & 50.8   & 50.3   & 2.4 & 52.0  & 57.1   \\
            6  & 63.8   & 62.1   & 6.6 & 70.5  & 80.6   \\ \midrule
            10 & 107.1 & 105.2 & 5.0 & 116.4 & 120.0 \\ \bottomrule
    \end{tabular}%
    % }
    \caption{Latencies of \texttt{dbt} command with varying call depths.}

    \label{table:bt_remote_latency}
\end{table}

\begin{table}[t]
    \centering
    % \footnotesize
    \small
    % \resizebox{\columnwidth}{!}{%
    \begin{tabular}{@{}lrr@{}}
        \toprule
        \textbf{Time APIs}                                           &
        \begin{tabular}[t]{@{}r@{}}\textbf{w/ PET} (ns)\end{tabular} &
        \begin{tabular}[t]{@{}r@{}}\textbf{w/o PET} (ns)\end{tabular}          \\
        \midrule
        \texttt{chrono::system\_clock}                               & 91 & 30 \\
        \texttt{std::time}                             & 86               & 4                \\
        \texttt{gettimeofday}                                        & 95 & 29 \\ \midrule
        \texttt{clock\_gettime(REALTIME)}                            & 88 & 28 \\
        \texttt{clock\_gettime(MONOTONIC)}                           & 89 & 28 \\
        % \texttt{clock\_gettime(MONOTONIC\_RAW)} & 88               & 28                \\
        \bottomrule
    \end{tabular}%
    % }
    \caption{Average execution time for time APIs with or without the PET shim layer interception.}
    \label{tab:faketime_perf}
\end{table}

\begin{table*}[!t]
    \small
    \centering
    \caption*{\small \textbf{Legend:} \hspace{1mm} \circledplain{1} Harm Performance \hspace{1mm} \circledplain{2} Disrupt System States \hspace{1mm} \circledplain{3} Change System Behaviors}
    \begin{tabular}{@{}l l N@{\hspace{0.3em}} U l | c | c @{}}
        \toprule
        \textbf{System} & \textbf{Case}           & \multicolumn{2}{c}{\textbf{Threshold}} & \textbf{Aftermath} & \textbf{Effect Category}            & \textbf{Severity}                                                 \\ \midrule
        \multirow[t]{4}{*}{LogCabin~\cite{LogcabinLogcabin2024}}
                        & Heartbeat Timeout       & 250                                    & ms                 & Send excessive heartbeats           & \circledplain{1}                                       & low      \\
        % & CondVar Deadline & 500 & ms & Immature thread wakeups & \circledplain{1} & moderate \\
                        & Election Timeout        & 500                                    & ms                 & Disrupt established leadership      & \circledplain{1} + \circledplain{2}                    & moderate \\
                        & Client Session Expiry   & 3600                                   & s                  & Disconnect clients undesirably      & \circledplain{2} + \circledplain{3}                    & severe   \\
        \midrule
        \multirow[t]{4}{*}{RAMCloud~\cite{ousterhoutRAMCloudStorageSystem2015}}
                        & RPC Timeout             & 100\,-\,200                            & ms \(^*\)          & Send excessive RPCs                 & \circledplain{1}                                       & low      \\
                        & Lease Renewal Threshold & 900                                    & s                  & Force clients to renew the lease    & \circledplain{2}                                       & moderate \\
                        & Transaction Timeout     & \(N \times 50\)                        & ms \(^\dagger\)    & Abort the transaction               & \circledplain{1} + \circledplain{2} + \circledplain{3} & severe   \\
                        & Crash Recovery Trigger  & 250                                    & ms                 & Initiate large-scale crash recovery & \circledplain{1} + \circledplain{2} + \circledplain{3} & severe   \\
        \midrule
        % \bottomrule
        \multirow[t]{5}{*}{Nu~\cite{ruanNuAchievingMicrosecondScale2023}}
        % & Proclet sort interval   & 200                                    & ms                 & Recalculate proclet priorities      & \circledplain{2} + \circledplain{3}                    & moderate \\
        % & RPC Batch Timeout & 5 & $\mu$s & Force RPC batch transmission & \circledplain{1} & moderate \\
                        & Full Shard Probing      & 400                                    & ms                 & Probe full shards for freed memory  & \circledplain{1} + \circledplain{2}                    & moderate \\
                        & Connection Pool Timeout & 10                                     & s                  & Connections considered dead & \circledplain{1} + \circledplain{2} + \circledplain{3} & severe   \\
                        & Service Keepalive       & 10                                     & s                  & Connections considered idle & \circledplain{2} + \circledplain{3}                    & severe   \\ \bottomrule
    \end{tabular}%
    \caption*{\small \(^*\) Value is chosen randomly between 100\,-\,200\,ms. \hspace{1mm} \(^\dagger\) \(N\) is the number of participants involved in the transaction.}
    \vspace{-0.4em}
    \caption{Timeout thresholds used in real-world systems, with their aftermath, effect category, and severity.}
    \label{tab:survey_time}
\end{table*}

\parab{Overhead.}
We first measure the overhead of PET's interception by invoking common POSIX time APIs one million times.
Compared to direct bare-metal execution, the intercepted time APIs introduce only nanosecond-scale overhead per call as described in~\tabref{tab:faketime_perf}.
Furthermore, the routine that calculates and publishes the cumulative pause offset before a process resumes execution takes 1.265\,$\mu$s on average.
These negligible overheads confirm that the PET layer does not distort the normal execution performance of the debugged application.
\looseness=-1

\parab{Masking Time Jumps.}
To evaluate the effectiveness of PET at masking time jumps, we run a synthetic application and repeatedly retrieve timestamps using POSIX time APIs when PET is applied.
The synthetic application executes a tight loop with a light workload ($\approx$\,50~$\mu$s) in each iteration.
We attach \sysname to the synthetic application with PET enabled and repeatedly inject pause and resume commands while the application is running.
In this experiment, we fix the interval between pause and resume at 100~ms; during real-world interactive debugging, this interval can be arbitrarily large.
Because we tightly control the interval of each loop iteration to the microsecond scale, any noticeable time jumps at the millisecond scale are effectively captured.
We present the timeline perceived by the synthetic application in \figref{fig:time-gap}.
The application consistently perceives $\leq$~5~ms time jumps, which is clearly smaller than the 100~ms interval.
Because PET accumulates and subtracts the cumulative pause offset, we observe similar sub-5~ms results regardless of the actual duration between pause and resume.

\parab{Real-World Safety Margin.}
To contextualize this 5~ms maximum drift, we examined default timing thresholds in three distributed systems, LogCabin~\cite{LogcabinLogcabin2024}, RAMCloud~\cite{ousterhoutRAMCloudStorageSystem2015}, and Nu~\cite{ruanNuAchievingMicrosecondScale2023}.
As summarized in~\tabref{tab:survey_time}, these systems rely on timeouts to detect and handle failures, such as LogCabin's election timers and RAMCloud's fast crash recovery~\cite{ongaroFastCrashRecovery2011}.
Exceeding these thresholds triggers non-negligible side effects, ranging from degraded performance via excessive heartbeats to disrupted system states via forced leader step-downs or a massive crash recovery.
Our survey indicates that the most aggressive timeout thresholds in these systems typically range from 50~ms to 250~ms.
Because PET reliably bounds the perceived time jump to under 5~ms, it provides an order-of-magnitude safety envelope that isolates the application from debugger-induced timeout cascades.

\begin{figure}
    \centering
    \includegraphics[width=0.98\linewidth]{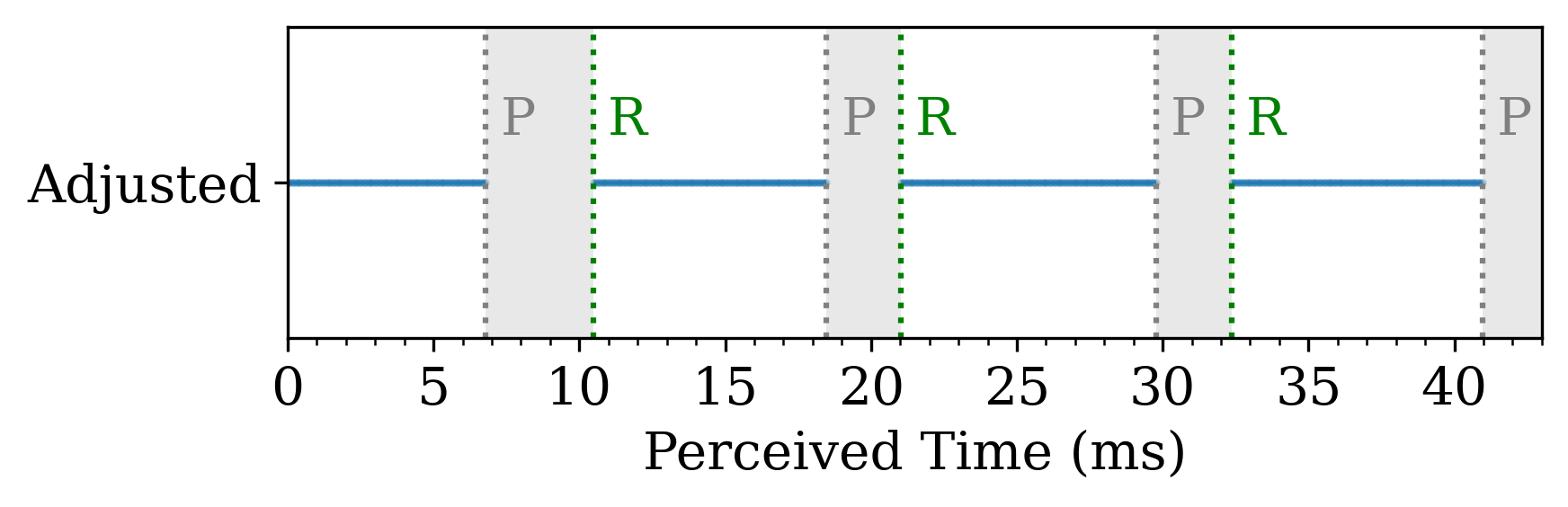}
    \caption{The application-perceived time (PET) with multiple pauses (\textcolor{gray}{P}) and resumes (\textcolor{darkgreen}{R}). The gray area shows the time jumps that the application perceives between an execution pause and resume.}
    \label{fig:time-gap}
\end{figure}

\subsection{DDB Runtime Overhead on Applications}
\label{sec:ddb_overhead}

While \sysname targets staging and test environments, its runtime impact must be comparable to a traditional single-process source-level debugger and incur no substantial overhead from the orchestration and management of distributed processes.
\figref{fig:ddb-overhead-tput} plots the retained throughput of the \texttt{socialnet} benchmark (on both the Nu and ServiceWeaver runtimes) and a gRPC-based Raft consensus implementation with and without \sysname attached.

\sysname introduces minimal throughput degradation on \texttt{socialnet}.
Specifically, \texttt{socialnet} running on ServiceWeaver and Nu exhibits a 0.6\% and 3.8\% throughput drop, respectively.
To establish a baseline for binary-level debugging overhead, we also compare \sysname against GDB.
Attaching GDB to the gRPC Raft service yields an 18.4\% throughput degradation.
Attaching \sysname to the same gRPC Raft service yields a 20.1\% throughput degradation.
This demonstrates that \sysname achieves distributed debugging capabilities with overhead comparable to a traditional single-process debugger.
The caller-context metadata embedding required for \texttt{dbt} adds only 52 bytes to each RPC payload, resulting in a minimal impact on network utilization.
\looseness=-1

\begin{figure}
    \centering
    \includegraphics[width=\columnwidth]{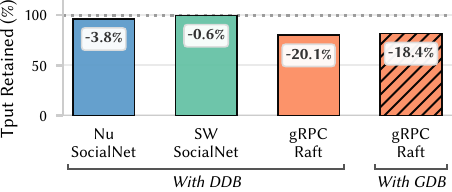}
    \caption{The performance overhead when DDB metadata is embedded, and DDB is attached across Nu's \texttt{socialnet}, ServiceWeaver (SW)'s \texttt{socialnet}, and gRPC-based Raft consensus implementation.}
    \label{fig:ddb-overhead-tput}
\end{figure}

\section{Related Work}

\parab{Formal methods and model checking.} Formal verification and model checking~\cite{wilcoxVerdiFrameworkImplementing2015, newcombeHowAmazonWeb2015, hawblitzelIronFleetProvingPractical2015} provide strong correctness guarantees for protocol specifications but operate on mathematical models, not running code.
Verification-guided tests are effective~\cite{wangModelCheckingGuided2023, yangmodisttransparentmodel2009}, but they require substantial expertise.
Runtime verification systems~\cite{alpernasCloudScaleRuntimeVerification2021, liuD3SDebuggingDeployed2008} check specifications against live executions but do not enable interactive inspection or modification of the state that triggered a violation.
\looseness=-1

\parab{Record and Replay.} Capturing and reproducing non-determinism are crucial in distributed applications~\cite{kharbandaAlwaysRecordingFrameworkServerless2023, ronsseRecPlayFullyIntegrated1999, zuoExecutionReconstructionHarnessing2021}. 
Some work~\cite{geelsReplayDebuggingDistributed2006, geelsFridayGlobalComprehension2007} also explores replay interfaces that resemble a traditional symbolic debugger, \eg GDB.
Record-and-replay is orthogonal to \sysname: it targets non-deterministic bugs by replaying a recorded execution, whereas \sysname operates on live, real-time execution.
\looseness=-1

\parab{Distributed logging, tracing, and monitoring.} These frameworks enhance the logging and reasoning experience by automatically aggregating logs~\cite{azuremonitorlogs2024, CloudLoggingGoogle, LogDeviceDistributedData2017}, establishing causality between cross-service interactions and improving observability~\cite{ashokTraceWeaverDistributedRequest2024, shenNetworkCentricDistributedTracing2023, macePivotTracingDynamic2015, openzipkindistributedtracing, jaeger, sigelmandapper, OpenTelemetry}, and enabling tracing-based aftermath-analysis to reveal erratic symptoms~\cite{zhangbenefithindsighttracing2023, huangTprofPerformanceProfiling2021, kaldorCanopyEndtoEndPerformance2017, gan2021sage}. Monitoring tools~\cite{grafanaopencomposable, rabensteinprometheusnextgenerationmonitoring2015} can also provide high-level metrics to indicate potential issues.
While these tools are effective at identifying \textit{which} service is behaving anomalously or reconstructing the sequence of events around a failure~\cite{dogga2021revelio, wittkopp2024logrca}, they do not provide source-level interactive debugging.
\sysname addresses a different need: once a developer knows \textit{where} a problem occurs, they need to understand \textit{why} by pausing execution, navigating the cross-service call chain, and inspecting live runtime state.
Unlike tracing, \sysname's metadata is not sampled and produces no persistent records; it is interrogated on demand.

\parab{HPC/Distributed debuggers.} TotalView~\cite{totalviewdebuggerhpc} and Linaro DDT~\cite{linaroddt} are interactive debuggers for HPC applications that support multi-process debugging at scale.
Although they provide mechanisms for managing breakpoints across MPI processes, they are designed for HPC's homogeneous execution model: they do not reconstruct call stacks across RPC boundaries, and they offer no protection against timeout interference from debugger pauses~\cite{zhou2021fault}. Both are also proprietary.
Ray~\cite{moritzraydistributedframework2018} features a built-in interactive source-level debugger~\cite{raydebugger} that provides a single-point view of distributed tasks within the Ray framework.
However, Ray's debugger does not offer distributed backtrace and still requires developers to manually cross-reference call stacks across process boundaries.
\looseness=-1

\parab{Classic distributed debugging.}
Early research in distributed debugging focused heavily on algorithms for consistent global snapshots~\cite{chandy1985distributed}, distributed halting~\cite{miller1988breakpoints}, and determining causal ordering of asynchronous message passing~\cite{fowler1990causal}.
However, modern cloud-native architectures have shifted the problem space.
Today's systems are characterized by deep RPC chains~\cite{sigelmandapper}, dynamically scaling replica sets~\cite{gan2021sage}, and strict timeout-driven availability requirements~\cite{dean2013tail}.
Where classic systems focused on theoretically consistent state spaces for asynchronous events, \sysname addresses practical developer workflow bottlenecks by providing an intent-preserving control plane that manages debug state across ephemeral infrastructure.

\section{Discussion}

\parab{Debugger backend.}
As of this writing, \sysname has been thoroughly tested with GDB, and we have implemented experimental support for LLDB.
Adding a new backend primarily requires \sysname to parse and interpret the backend's commands, responses, and execution states through its native communication interface (\eg GDB's Machine Interface).
Because \sysname relies on a small set of custom debugger commands to implement distributed backtraces and pause-erased time, the target debugger must also provide sufficient extensibility, as GDB does through Python scripting.

\parab{Architecture assumptions.}
\sysname's distributed backtrace relies on architecture-specific register layouts and calling conventions to reconstruct call stacks.
The other functionality of \sysname is not tied to a specific architecture, provided that the debugger backend supports that architecture.
We have implemented and validated our stack-restoration approach on \texttt{aarch64} as well, with further implementation details in Appendix~\secref{app:impl-nits}.

\parab{OS assumptions.}
The current PET implementation assumes conventional kernel-level semantics for POSIX-compliant time APIs, so its correctness may depend on the target kernel: during prototyping, we found that several time APIs on macOS (Darwin) differ from their Linux counterparts despite sharing the same signatures.

\parab{When DDB helps and when it does not.}
DDB is particularly useful in early development and prototyping, where the instrumentation is itself the bottleneck, as our user study shows.
It complements logs and traces: once those surface a clue, the developer can inspect the suspected code paths from a chosen entrypoint, not by wading through large volumes of data.
Its closer counterpart is the iterative log-and-redeploy loop developers use to understand why a reproducible failure occurs.
Reproducing distributed failures is genuinely hard, and \sysname does not try to solve it; for non-reproducible production incidents in hyperscaler-operated systems, observability tools and record-replay remain the right tools.
\section{Conclusion}

Interactive debugging enables developers to pause execution, inspect runtime state, and navigate call stacks at the source level, yet it has been widely treated as impractical for distributed systems.
We present \sysname, a distributed interactive debugger that restores this tight hypothesis-testing loop through three user-space mechanisms---Distributed Backtrace, an intent-preserving control plane, and Pause-Erased Time---without requiring kernel modifications.

In a controlled user study, \sysname achieves 100\% fault localization, compared to 38.5\% for baseline tools.
It delivers 30~ms median backtrace latency and 1--5\% throughput overhead, and requires only 10--60 lines of per-framework integration.

\begin{acks}
We thank our shepherd and reviewers for their invaluable feedback. 
We thank the members of the USC Networked Systems Lab (NSL) for their insightful discussions.
We appreciate the participants of the user study for their time.
Results presented in this paper were partially obtained using the Chameleon and CloudLab testbeds supported by the National Science Foundation.
This work was funded in part by a National Research Foundation of Korea (NRF) grant (RS-2024-00354947).
\end{acks}

%% ============================================================
%% Bibliography (auto-selects the right style)
%% ============================================================
\clearpage
\newpage
\ifusenix
  \bibliographystyle{plain}
\else
  \bibliographystyle{templates/acmart/ACM-Reference-Format-No-Year}
\fi
\bibliography{1_main,references}

\clearpage
\appendix

\section*{Appendix}

\noindent
Following the practice at SOSP'26, we declare that the appendix and included materials were not peer reviewed by the SOSP committee.

\section{PET Mechanism}
\label{app:pet}

\subsection{Formal Description and Invariants Correctness}
\label{app:pet-prove}

We formalize the Pause-Erased Time (PET) abstraction and prove two invariants that guarantee PET induces no incorrect timing behavior.
Let $t \in \mathbb{R}_{\ge0}$ denote the continuous absolute physical time.
Suppose the debugger suspends the target process multiple times during a debugging session.
We represent these debugger-induced pauses as a sequence of disjoint, half-open physical time intervals $G_k = [p_k, r_k)$.
Here, $p_k$ and $r_k$ are the exact physical times on the $t$ continuum when the process pauses and resumes, respectively.
During each pause interval $G_k$, the process executes no user code and makes no time API calls.
To measure the duration of these pauses from within the process, we rely on the operating system's timekeeping.
Let $T(t)$ denote the value of the kernel's monotonic clock at any given physical time $t$.
The shim layer measures the duration of each pause $G_k$ by recording two internal timestamps from this kernel clock $T$.
Specifically, it records $t_{\text{stop}}^{(k)}$ immediately after all threads are stopped, and $t_{\text{resume}}^{(k)}$ immediately before execution resumes.
Because these internal timestamp reads inevitably occur sequentially within the physical pause interval $[p_k, r_k)$, their recorded values are strictly bounded by the true kernel clock values at the interval's boundaries: $T(p_k) \le t_{\text{stop}}^{(k)} < t_{\text{resume}}^{(k)} \le T(r_k)$.
The shim layer defines the measured pause duration $\Delta_k$ for the interval $G_k$ as the difference between these two internal timestamps.
\[
\Delta_k \;=\; t_{\text{resume}}^{(k)} - t_{\text{stop}}^{(k)} \;\in\; [0,\; T(r_k) - T(p_k)].
\]
As the process resumes at physical time $r_k$, the shim layer adds $\Delta_k$ to a global cumulative offset $O(t)$.
Formally, we define $O(t) = \sum_{k:\, r_k\le t} \Delta_k$.
When the application subsequently calls a time API, the shim layer intercepts the call and returns the virtualized time $T_v(t) = T(t) - O(t)$.
This mechanism implements three time-related semantics:
\begin{inparaenum}[(i)]
    \item \texttt{get\_time} directly returns $T_v(t)$.
    \item \texttt{sleep\_until}$(D_v)$ safely blocks the thread and returns at the earliest physical time $t^\star$ such that $T_v(t^\star) \ge D_v$.
    To achieve this, the shim layer repeatedly arms the kernel's blocking system call with an absolute physical deadline $D_{\text{real}} = D_v + O(t_{\text{now}})$.
    Whenever the kernel wakes the thread prematurely---either due to the offset $O$ increasing from an intervening pause or a spurious signal---the shim layer intercepts the wakeup, recalculates $D_{\text{real}}$, and re-arms the sleep until the virtual condition holds.
    \item \texttt{sleep}$(D)$ is syntactic sugar for sleeping until a relative duration elapses, implemented as \texttt{sleep\_until}$(T_v(t_{\text{now}}) + D)$.
\end{inparaenum}

\begin{lemma}[Cumulative offset property]\label{lem:offset-diff}
For any physical times $t_1 < t_2$, the change in the cumulative offset is exactly the sum of measured pause durations that concluded within that interval.
Specifically, $O(t_2) - O(t_1) = \sum_{k:\, t_1 < r_k \le t_2} \Delta_k$.
\end{lemma}

\begin{proof}
By definition, $O(t)$ is a right-continuous step function that only increases by $\Delta_k$ precisely at $t = r_k$.
The function otherwise remains constant during all running intervals and pauses.
Therefore, the difference $O(t_2) - O(t_1)$ perfectly isolates and sums the increments of those steps where $r_k \in (t_1, t_2]$.
\end{proof}

\parae{Invariant 1: Monotonicity.}
\label{thm:monotonic}
The Pause-Erased Time never decreases.
Let two \texttt{get\_time} calls occur at physical times $t_1 < t_2$, returning $v_1 = T_v(t_1)$ and $v_2 = T_v(t_2)$.
Then $v_2 \ge v_1$.

\begin{proof}
Because the application is entirely suspended during any pause interval $G_k$, the observation times $t_1$ and $t_2$ must strictly fall within running intervals.
Consequently, any pause interval $G_k$ that intersects the window $(t_1, t_2]$ must be fully contained within it, implying $t_1 < p_k < r_k \le t_2$.
We can partition the timeframe $(t_1, t_2]$ into a collection of maximal running intervals $\mathcal{R}$ and a collection of paused intervals $\mathcal{P} = \{G_k \mid G_k \subseteq (t_1, t_2]\}$.
The total elapsed real time measured by the kernel clock decomposes over these intervals.
\[
\begin{aligned}
T(t_2) - T(t_1) 
    &= \sum_{I \in \mathcal{R}} \bigl[T(\max I) - T(\min I)\bigr] \\
    &\quad + \sum_{G_k \in \mathcal{P}} \bigl[T(r_k) - T(p_k)\bigr]
\end{aligned}
\]
We expand the difference in virtual time using Lemma~\ref{lem:offset-diff}.
\[
\begin{aligned}
v_2 - v_1 
    &= \bigl(T(t_2) - T(t_1)\bigr) - \bigl(O(t_2) - O(t_1)\bigr) \\
    &= \sum_{I \in \mathcal{R}} \bigl[T(\max I) - T(\min I)\bigr] \\
    &\quad + \sum_{G_k \in \mathcal{P}} \bigl[T(r_k) - T(p_k) - \Delta_k\bigr]
\end{aligned}
\]
The kernel's monotonic clock strictly non-decreases, making the first sum over the running intervals naturally non-negative.
For the second sum, earlier definitions establish that the measured pause $\Delta_k$ is bounded by the true interval duration, meaning $\Delta_k \le T(r_k) - T(p_k)$.
Because both sums are non-negative, the virtual time successfully preserves monotonicity such that $v_2 - v_1 \ge 0$.
\end{proof}

\parae{Invariant 2: Timer Correctness.}
\label{thm:timer-correctness}
Time-blocking operations unblock strictly when their virtual time constraints are satisfied.
Specifically, a call to \texttt{sleep\_until}$(D_v)$ returns at the earliest physical time $t^\star$ where $T_v(t^\star) \ge D_v$.
Consequently, a relative \texttt{sleep}$(D)$ issued at physical time $t_0$ returns at an unblocking instant $t^\star$ that guarantees $T_v(t^\star) - T_v(t_0) \ge D$.

\begin{proof}
When an application blocks using \texttt{sleep\_until}$(D_v)$, the shim layer arms the underlying kernel sleep primitive with the absolute physical deadline $D_{\text{real}} = D_v + O(t_{\text{now}})$.
The shim layer subsequently intercepts every wakeup.
It unconditionally returns control to the application only when the virtual time satisfies $T_v(t) \ge D_v$.
We must guarantee that the armed physical deadline $D_{\text{real}}$ mathematically expires no later than the true physical time when the virtual condition is satisfied.
Suppose at some subsequent physical time $t_{\text{target}}$, the virtual time reaches the deadline such that $T_v(t_{\text{target}}) \ge D_v$.
By definition, this implies $T(t_{\text{target}}) - O(t_{\text{target}}) \ge D_v$, which rearranges algebraically to $T(t_{\text{target}}) \ge D_v + O(t_{\text{target}})$.
Because the cumulative offset never decreases, we have $O(t_{\text{target}}) \ge O(t_{\text{now}})$.
Therefore, we are guaranteed that $T(t_{\text{target}}) \ge D_v + O(t_{\text{now}}) = D_{\text{real}}$.
This inequality is essential: it proves that the initially armed kernel deadline $D_{\text{real}}$ will inevitably expire at or strictly prior to any physical time where the virtual deadline is satisfied.
If no intervening pauses occur, the offset remains unchanged and the kernel wakes the thread exactly when the valid condition $T_v(t) \ge D_v$ is met.
However, if an intervening pause occurs, the offset increases, causing $D_{\text{real}}$ to expire prematurely while the actual virtual time $T_v(t)$ remains strictly less than $D_v$.
The shim layer automatically traps this premature wakeup, recalculates a pushed-back $D_{\text{real}}$ using the updated offset, and re-arms the sleep.
By repeatedly filtering out these premature wakeups through persistent recalculation, the shim layer guarantees that the application ultimately unblocks precisely at the earliest observable physical instant satisfying the virtual deadline.

For a relative \texttt{sleep}$(D)$, the implementation uniformly maps the request to an absolute wait using \texttt{sleep\_until}$(D_v)$, where the deadline is initialized to $D_v = T_v(t_0) + D$.
By the preceding logic, the unblocking instant $t^\star$ guarantees $T_v(t^\star) \ge D_v$.
Substituting $D_v$, we conclude that $T_v(t^\star) - T_v(t_0) \ge D$.
Thus, relative sleep durations are preserved across arbitrary debugger pauses.
\end{proof}

\parab{Guarantees and Scope.} 
The proof of these two invariants substantiates the mathematical soundness of the pause-erased time semantics.
Because arbitrarily complex timeout and rate-limiting schemes ultimately decompose into these elemental time-reads and blocking sleeps, the shim layer averts any artificial failure cascades induced by debugger pauses.
System decisions evaluated against PET, such as lease expiration and election deadlines, remain observationally equivalent to an unpaused execution state.
This temporal equivalence holds completely up to a minimal under-measurement slack, formally defined in~\secref{app:under-measurement} and empirically verified in~\secref{sec:time-jump}.
The PET abstraction specifically governs local time reads and localized timeout evaluations, making no attempts to artificially synchronize or warp clocks across divergent system processes.

\subsection{PET Under-measurement Slack}
\label{app:under-measurement}

The Pause-Erased Time mechanism cannot entirely eliminate the physical duration of execution pauses due to unavoidable implementation latencies across the suspension and resumption pathways.
Recall from the definitions above that for any debugger-induced pause $G_k = [p_k, r_k)$, the true physical time elapsed during the suspension is exactly the kernel clock difference $\Delta^{\text{true}}_k = T(r_k) - T(p_k)$.
To ensure the system never over-estimates the pause duration---which would dangerously cause virtual time to flow backwards---the measurement is taken strictly from within the interval boundaries.
Specifically, the internal timestamps $t_{\text{stop}}^{(k)}$ and $t_{\text{resume}}^{(k)}$ are captured exclusively after all debuggee threads are safely paused and exclusively before they are allowed to resume user-space execution.

Because coordinating thread suspension across a process involves asynchronous signals and OS scheduler delays, there exists a strictly non-negative latency at both boundaries.
The measured pause duration $\Delta_k = t_{\text{resume}}^{(k)} - t_{\text{stop}}^{(k)}$ is therefore strictly less than or equal to the true pause duration.

\begin{remark}[Under-measurement slack]
Let $\eta_k$ denote the under-measurement slack for a pause interval $G_k$, defined mathematically as the unmeasured residual duration:
\[
\eta_k \;\stackrel{\mathrm{def}}{=}\; \bigl(T(r_k) - T(p_k)\bigr) \;-\; \Delta_k \;\ge\; 0.
\]
Let $T^{\text{ideal}}_v(t)$ formally represent a hypothetically perfect virtual time that completely erases the true physical duration of every pause without any latency:
\[
T^{\text{ideal}}_v(t) \;\stackrel{\mathrm{def}}{=}\; T(t) \;-\; \sum_{k:\,r_k\le t}\bigl(T(r_k) - T(p_k)\bigr).
\]
For any physical time $t$, the actual virtual time $T_v(t)$ governed by PET algebraically relates to the theoretically ideal virtual time as follows:
\[
T_v(t) \;=\; T^{\text{ideal}}_v(t) \;+\; \sum_{k:\,r_k\le t} \eta_k.
\]
This mathematically demonstrates that the actual virtual time $T_v$ perpetually remains (weakly) ahead of the ideal debug time by exactly the global cumulative slack $\sum\eta_k$.
Furthermore, if the OS-level pause and resume measurement latencies are empirically bounded by maximum delays $\sigma_{\text{pause}}$ and $\sigma_{\text{resume}}$, the per-pause slack is strictly bounded by $\eta_k \le \sigma_{\text{pause}} + \sigma_{\text{resume}}$.
\end{remark}

\subsection{Wall Clock as the Clock-Source}
\label{app:wall-clock}

For APIs that use the \texttt{REALTIME} (wall clock) as the clock-source, Invariant 1 (Monotonicity) inherently does not apply, as NTP synchronization can abruptly step the physical wall clock backwards or forwards.
Crucially, this is a fundamental property of the OS wall clock, and handling such organic time jumps is strictly an application-level responsibility.
PET's formal correctness therefore remains intact: Invariant 2 (Timer Correctness) holds.
The shim layer guarantees that the virtualized wall clock yields an observationally equivalent unpaused execution, preserving the natural slews and steps of the physical wall clock while omitting the debugger pauses.

We reuse the established mathematical notation from the preceding proofs: continuous absolute physical time $t$, the measured cumulative pause offset $O(t) = \sum_{r_k\le t}\Delta_k$, and the debugger pause intervals $G_k = [p_k, r_k)$.
Let $R(t)$ denote the value of the kernel's wall clock at any given physical time $t$.
We define the virtualized wall clock value returned to the application as:
\[
R_v(t) \;\stackrel{\text{def}}{=}\; R(t) \;-\; O(t).
\]

Here, $R_v(t)$ acts as the true wall clock $R(t)$ observed on a pristine timeline where all recorded debugger pauses are seamlessly subtracted.
For any two physical time reads $t_1 < t_2$ (occurring outside of an execution pause), we observe:
\[
\begin{aligned}
R_v(t_2) - R_v(t_1) 
    &= \bigl[R(t_2) - R(t_1)\bigr] \;-\; \bigl[O(t_2) - O(t_1)\bigr] \\
    &= \bigl[R(t_2) - R(t_1)\bigr] \;-\; \sum_{t_1 < r_k \le t_2} \Delta_k.
\end{aligned}
\]
This proves that $R_v$ preserves the exact slews and discrete jumps of the underlying physical wall clock $R(t)$, omitting only the measured duration of any debugger-induced pause intervals.

\parab{REALTIME API Semantics.}
The shim layer maps the wall clock time APIs as follows:
\begin{inparaenum}[(i)]
    \item \texttt{get\_time\_realtime} directly returns $R_v(t)$.
    \item \texttt{sleep\_until\_realtime}$(D_v)$ safely blocks the thread by arming a \texttt{REALTIME} absolute physical deadline $D_{\text{real}}^{\mathrm{RT}} = D_v + O(t_{\text{now}})$.
    Whenever the kernel wakes the thread prematurely---due to the offset $O$ increasing from a pause, a spurious signal, or an NTP-induced backward clock step---the shim layer recalculates $D_{\text{real}}^{\mathrm{RT}}$ and re-arms the sleep until the virtual condition $R_v(t) \ge D_v$ properly holds.
    \item \texttt{sleep\_realtime}$(D)$ reduces to \texttt{sleep\_until\_realtime}$(R_v(t_{\text{now}}) + D)$.
\end{inparaenum}

\parae{Invariant 2: Timer Correctness for REALTIME.}
Time-blocking operations dependent on the wall clock unblock strictly when their virtual wall time constraints are satisfied.
Specifically, a call to \texttt{sleep\_until\_realtime}$(D_v)$ returns at an unblocking physical instant $t^\star$ where $R_v(t^\star) \ge D_v$.
Consequently, a relative \texttt{sleep\_realtime}$(D)$ issued at physical time $t_0$ returns at an instant $t^\star$ guaranteeing $R_v(t^\star) - R_v(t_0) \ge D$.
\looseness=-1

\begin{proof}
By construction, the shim layer yields control back to the application only when it verifies $R_v(t^\star) \ge D_v$, which occurs strictly when the underlying kernel reports $R(t^\star) \ge D_{\text{real}}^{\mathrm{RT}}$ using the most recently calculated deadline $D_{\text{real}}^{\mathrm{RT}} = D_v + O(t^\star)$.
If the physical wall clock $R(t)$ steps backwards due to NTP synchronization, the virtual wall clock $R_v(t)$ accurately reflects this backward jump.
Because the shim intercepts every kernel wakeup and re-evaluates the inequality $R_v(t) \ge D_v$, backward steps of $R(t)$ prolong the suspension exactly as they would in an un-debugged execution.
The re-arming loop masks both debugger pauses (which abruptly increase $O(t)$) and spurious interrupts, guaranteeing the earliest legitimate release time.
Similarly, because \texttt{sleep\_realtime}$(D)$ expands to an absolute deadline $D_v = R_v(t_0) + D$, the release at $t^\star$ enforces $R_v(t^\star) \ge R_v(t_0) + D$, preserving relative wall clock durations despite temporal discontinuities.
\looseness=-1
\end{proof}

\subsection{vDSO Handling}
\label{app:vdso}

Linux's vDSO (virtual Dynamic Shared Object) is a kernel-provided shared library mapped directly into a process's user-space memory, exposing critical kernel services without the overhead of a system call context switch.
Most modern \texttt{libc} implementations optimize time-reading APIs (\eg \texttt{gettimeofday} and \texttt{clock\_gettime}) by internally routing them through the vDSO.

Our implementation supports these vDSO optimizations without undermining their performance benefits.
The shim layer interposes itself strictly between the user application and the underlying \texttt{libc}.
When an application invokes a time API, the symbol is intercepted by the shim.
The shim then queries the actual \texttt{libc} function---which uses the vDSO---to retrieve the true physical time, and applies the offset adjustments.
Consequently, our design inherits the zero-syscall performance of vDSO while enforcing the Pause-Erased Time abstraction.
\looseness=-1

We acknowledge a structural limitation: if an application bypasses \texttt{libc} to invoke function pointers directly from the vDSO memory mapping (or uses raw \texttt{syscall} assembly), our \texttt{LD\_PRELOAD} shim cannot intercept the observation.
Such raw-level resolution is rare in application code, but it does break the shim's enforcement scope.
However, because PET's semantics are defined by the invariants above rather than by the interception mechanism, the same logic can be moved into deeper interception frameworks (\eg eBPF or hypervisor trapping) if required.

\section{Distributed Backtrace}

\subsection{Distributed Stack Stitching}
\label{app:dss}

\begin{algorithm}[t]
\small % Added to reduce the font size of the algorithm content
\SetKwInOut{Input}{input}
\SetKwInOut{Output}{output}
\Input{thread\_id}
\Output{stitched call stack across processes}

\BlankLine
\SetKwFunction{FCross}{DistributedBacktrace}
\FCross{thread\_id}\\
\While{true}{
    stacks $\gets$ FetchStack(thread\_id)\;
    meta $\gets$ ScanMetadata(stack)\;
    
    \eIf{meta is invalid}{
        \Return{StitchAllStacks(stacks)}\;
    }{
        caller $\gets$ GetCaller(meta["caller\_id"])\; 
        RestoreStack(caller, meta["caller\_context"])\;
    }
    
    thread\_id $\gets$ caller.thread\_id\;
}
\caption{Distributed stack stitching.}
\label{list:cross-machine-backtrace}
\end{algorithm}
As described in~\algoref{list:cross-machine-backtrace}, \sysname reconstructs the global call graph by iteratively tracing the causal chain backward across network boundaries.
Specifically, \sysname:
\begin{inparaenum}[(1)]
    \item fetches the local call stack of the targeted thread and appends it to the stitched array;
    \item unwinds the stack to locate the injected RPC frame and parses the embedded causal metadata;
    \item terminates the local trace if no valid metadata is found, concluding that the root of the call chain has been reached;
    \item utilizes the extracted metadata to identify the upstream caller process (which may reside on a remote node) and restores the caller's exact register-level execution context;
    \item transitions the backtrace focus to the restored upstream thread and recursively repeats the entire trace process until the origin is fully resolved.
\end{inparaenum}

\subsection{Stack Stitching Latency at Scale} 
\label{sec:dbt-latency}
Because the distributed backtrace algorithm recursively sweeps the call stack backwards across RPC boundaries, its time complexity is inherently linear with respect to the depth of the distributed call chain.
Consequently, resolving full traces for deeply nested call chains incurs proportional end-to-end cross-RPC latency.
A viable optimization is to propagate $N$ consecutive ancestor metadata records within each RPC packet, rather than just the immediate parent.
\sysname can then dispatch concurrent stitching queries to multiple ancestor nodes in the chain, masking the sequential round-trip delays.
This optimization trades off against a modest increase in the RPC payload size during routine application execution.
\looseness=-1

\section{Implementation Details}
\label{app:impl-nits}

\parab{Framework Integration.} Integrating \sysname into existing distributed frameworks requires minimal engineering effort. For high-level orchestrators like Nu, Quicksand, and ServiceWeaver, \sysname's instrumentation is embedded entirely within the framework's underlying runtime, achieving complete transparency for end-user application code. For standard RPC libraries like gRPC, minor boilerplate adjustments during service initialization are sufficient to propagate discovery metadata.

\parab{Metadata Extraction Barrier.} During an RPC invocation, \sysname injects causal metadata into a local stack variable to mark the network boundary. Because this variable is read only asynchronously, by an external debugger process at runtime, optimizing compilers treat the assignment as dead code and elide it during compilation. To circumvent this, we insert a \texttt{volatile} inline assembly directive as an optimization barrier, forcing the compiler to materialize and preserve the metadata on the physical stack.

\parab{Wait-for-Attach.} To permit deterministic debugging of early startup routines, \sysname provides a wait-for-attach mechanism. When enabled via the \texttt{DConnector} module, the application explicitly blocks its initialization path and waits for an external attachment signal. We repurpose the POSIX real-time signal \texttt{SIGRTMIN+6} (\texttt{SIG40}) as a dedicated wakeup trap. Once the debugging session is established, the central control plane instructs the local \texttt{DKnot} daemon to inject \texttt{SIG40} into the debuggee, resuming its execution.

\parab{Multi-ABI Stack Restoration.} Properly stitching causal call stacks dynamically requires rewriting explicit thread hardware registers, strictly governed by the target's Application Binary Interface (ABI). For Go on \texttt{x86\_64}, overwriting the stack pointer (\texttt{RSP}) and instruction pointer (\texttt{RIP}) is natively sufficient. In contrast, for C++, the C ABI mandates additionally restoring the frame base pointer (\texttt{RBP}). Architectural boundaries present similar strict register requirements; ARM64 (\texttt{aarch64}) deployments intrinsically require precise restoration of the Link Register (\texttt{lr}) alongside \texttt{sp} and \texttt{pc} to guarantee that subsequent unwinding operations successfully chain the call stack.

\parab{Auxiliary Thread for Expression Evaluation.} To inspect complex application states, developers routinely evaluate in-guest functions. Dynamically executing arbitrary application logic inherently necessitates harnessing a live thread context. To prevent side effects or accidental state corruption in the primary application execution, we spawn an isolated, no-op auxiliary thread within each debuggee process dedicated to expression evaluation. During interactive evaluations, \texttt{DKnot} resumes only this auxiliary thread while all primary application threads remain frozen.

\parab{Debugger with Kubernetes.} Modern production container images are typically minimized (distroless) to reduce attack surfaces, making it impractical to bundle debugger binaries alongside microservices. For our ServiceWeaver evaluation over Kubernetes, \sysname injects its debugger agents into targeted, running pods using Kubernetes Ephemeral Containers~\cite{EphemeralContainers}. Because Kubernetes overlay networks are isolated, we deploy a gateway container that exposes a single multiplexed port to the \sysname control plane, bridging remote SSH tunnels into the ephemeral debugging containers attached to the target microservices.

% \section{Extended Timeout Threshold Case Study}
% \label{app:extended-case-study}
% \input{tables/case_study_appendix}

% In this section, we include an extended table (\tabref{tab:survey_time_extended}) and discussions for the case study of timeout threshold in real-world distributed applications and frameworks.

% \parab{Nu.}
% Nu is a distributed programming framework that provides transparent live migration for distributed applications. Built on top of the Caladan runtime, Nu enables proclets (process-lets) to migrate seamlessly across machines while
% maintaining microsecond-level latency. 

% The proclet sort interval determines how frequently Nu recalculates proclet priorities based on memory and CPU pressure---delays here can lead to suboptimal migration decisions and load imbalance.

% The full shard probing interval manages Nu's distributed memory pool. Without periodic probing, memory shards marked as full would never be rechecked, leading to memory fragmentation and allocation failures.

% At the application layer, Social Network services built with Nu developed  use 10-second timeouts for both connection pool acquisition and keepalives. If a client cannot obtain a connection from the pool within this timeout, the request fails and an error is returned, directly impacting user experience. Keepalive timeouts handles idle connections, which can disrupt service mesh communication patterns and introduce the overhead of re-establishing connections.

\section{Extended Measurement of Command Handling Latencies}
\label{app:ext-cmd-handling}

\begin{table}[t]
\centering
\footnotesize

% --- Table 1 (38 processes) ---
\resizebox{\columnwidth}{!}{%
\begin{tabular}{@{}rrrrrr@{}}
\toprule
\multirow{2}{*}{\textbf{Command}} & \multicolumn{5}{c}{\textbf{Latency at the Scale of 38 Processes (ms)}} \\ \cmidrule(l){2-6}
 & \textbf{Count} & \textbf{Mean} & \textbf{Median} & \textbf{P95} & \textbf{P99} \\ \midrule
dbt              & 2,348 & 32.5 & 14.2 & 128.5 & 153.9 \\
continue         & 37    & 2.2  & 2.0  & 3.0   & 4.2   \\
break-delete     & 1     & 1.8  & 1.8  & 1.8   & 1.8   \\
break-insert     & 1     & 64.4 & 64.4 & 64.4  & 64.4  \\
exec-interrupt   & 4     & 8.1  & 7.3  & 7.8   & 7.8   \\ \midrule
stack-list-variables & 66 & 3.1  & 2.6  & 4.6   & 12.8  \\
var-create       & 146   & 24.4 & 7.4  & 48.5  & 51.2  \\
var-list-children& 3     & 3.0  & 3.0  & 3.0   & 3.0   \\
var-update       & 11    & 19.2 & 4.6  & 47.6  & 47.6  \\ \bottomrule
\end{tabular}%
}
\caption{Measured latencies of commands at the scale of 38 processes.}
\label{table:full-cmd-lat-38}

\vspace{0.8em} % spacing between tables

% --- Table 2 (62 processes) ---
\resizebox{\columnwidth}{!}{%
\begin{tabular}{@{}rrrrrr@{}}
\toprule
\multirow{2}{*}{\textbf{Command}} & \multicolumn{5}{c}{\textbf{Latency at the Scale of 62 Processes (ms)}} \\ \cmidrule(l){2-6}
 & \textbf{Count} & \textbf{Mean} & \textbf{Median} & \textbf{P95} & \textbf{P99} \\ \midrule
dbt              & 4,303 & 32.1 & 13.9 & 144.5 & 156.3 \\
continue         & 8     & 2.8  & 2.7  & 3.6   & 3.6   \\
break-delete     & 1     & 6.9  & 6.9  & 6.9   & 6.9   \\
break-insert     & 2     & 46.9 & 46.9 & 8.5   & 8.5   \\
exec-interrupt   & 4     & 15.4 & 16.1 & 17.3  & 17.3  \\ \midrule
stack-list-variables & 28 & 6.1  & 3.8  & 12.9  & 30.3  \\
var-create       & 78    & 24.6 & 16.3 & 49.6  & 50.1  \\
var-list-children& 3     & 5.6  & 3.8  & 3.8   & 3.8   \\ \bottomrule
\end{tabular}%
}
\caption{Measured latencies of commands at the scale of 62 processes.}
\label{table:full-cmd-lat-62}

\vspace{0.8em}

% --- Table 3 (122 processes) ---
\resizebox{\columnwidth}{!}{%
\begin{tabular}{@{}rrrrrr@{}}
\toprule
\multirow{2}{*}{\textbf{Command}} & \multicolumn{5}{c}{\textbf{Latency at the Scale of 122 Processes (ms)}} \\ \cmidrule(l){2-6}
 & \textbf{Count} & \textbf{Mean} & \textbf{Median} & \textbf{P95} & \textbf{P99} \\ \midrule
dbt              & 14,573 & 26.5 & 13.3 & 110.1 & 160.1 \\
continue         & 7      & 4.7  & 4.8  & 4.9   & 4.9   \\
break-insert     & 1      & 81.1 & 81.1 & 81.1  & 81.1  \\
exec-interrupt   & 6      & 26.1 & 25.6 & 29.0  & 29.0  \\ \midrule
stack-list-variables & 19 & 7.6  & 4.4  & 15.7  & 15.7  \\
var-create       & 43     & 24.4 & 18.8 & 49.8  & 69.9  \\ \bottomrule
\end{tabular}%
}
\caption{Measured latencies of commands at the scale of 122 processes.}
\label{table:full-cmd-lat-122}

\end{table}

To complement the primary evaluation, we present the full measurements of \sysname's command handling latencies across varying deployment scales (ServiceWeaver's \texttt{socialnet} app).
Specifically, \tabref{table:full-cmd-lat-38}, \tabref{table:full-cmd-lat-62}, and \tabref{table:full-cmd-lat-122} detail the granular end-to-end response times measured under settings of 38, 62, and 122 concurrent microservice instances.

As illustrated in the results, state-inspection commands are automatically emitted in bulk by the IDE frontend (\eg VS Code) to dynamically instantiate and fetch local variable scopes for UI rendering.
Examples include \texttt{stack-list-variables}, \texttt{var-create}, \texttt{var-list-children}, and \texttt{var-update}.
Conversely, the \texttt{exec-interrupt} command operates as an asynchronous control signal, issued exclusively when a developer manually asserts a global pause sequence via the debugger control panel.

\section{Migration Continuity (Nu and Quicksand)}
\label{app:heap_restore}

\sysname's design guarantees extensibility to dynamic microsecond-scale environments like Nu~\cite{ruanNuAchievingMicrosecondScale2023} and Quicksand~\cite{quicksand-nsdi25}. 
Specifically, \sysname preserves two forms of debugging continuity:
\begin{inparaenum}[(i)]
    \item when a \emph{computation migrates}, all established breakpoint intents automatically follow the execution context; and
    \item when \emph{heap state migrates}, the distributed backtrace reconstructs the global call graph with all caller-visible local variables intact.
\end{inparaenum}

\parab{Architectural Assumptions.}
To orchestrate heap continuity, \sysname leverages two core architectural properties natively provided by these advanced execution frameworks:
\begin{inparaenum}[(i)]
  \item \textit{Symmetric Address-Space Layouts.} The underlying frameworks deploy an identical, universally linked binary across all cluster nodes with Address Space Layout Randomization (ASLR) strictly disabled. 
  Consequently, virtual memory addresses and object layouts are guaranteed identical cluster-wide. 
  This implies that \sysname's inspection-time restoration can safely map distributed objects directly to the same virtual addresses without requiring complex pointer swizzling or relocation maps.
  \item \textit{Logical Heap Locator API.} The frameworks expose a queryable API that, given a stable logical context identifier, returns the specific remote host and process currently owning a migrating heap region. \sysname invokes this locator exclusively on-demand during a distributed backtrace.
\end{inparaenum}

\parab{Breakpoints Follow Computations.}
Environments like Nu and Quicksand enforce that any specific logical computation unit is resident on exactly one host at any given time.
\sysname decouples breakpoints from physical memory addresses, instead expressing them as logical \emph{intents} broadcast to all attached \texttt{DKnot} daemons.
As the orchestrator spawns new processes to handle migrating workloads, \sysname's service discovery captures them so they inherit the global breakpoint configurations.
Consequently, whichever physical process currently hosts the migrating computation reports the hit, and \sysname focuses the debugger on that entrypoint.
\looseness=-1

\parab{Reconstructing Stack Frames Post-Migration.}
Because Nu and Quicksand pin the originating sender thread on its native process while its outbound RPC remains outstanding, the stack frames remain sound even if surrounding threads or heaps relocate.
However, to reveal on-heap objects reached from a caller frame whose heap has already migrated, \sysname executes a sequence of \emph{on-demand heap restorations} during each traceback hop:

\begin{enumerate}[leftmargin=1.2em,itemsep=0.2em,topsep=0.2em]
  \item \textbf{Presence Verification.} For every reconstructed caller frame, \sysname queries the local runtime to verify whether the caller's target heap remains physically present in the local address space.
  \item \textbf{Surgical Heap Restoration.} If the heap is absent, \sysname queries the framework's locator API to pinpoint the current physical owner. It then coordinates a remote dump of the latest synchronized heap state and writes these bytes into the reconstructed caller's address space at the same virtual addresses, rendering all pointers and local variables valid for live IDE inspection.
  \item \textbf{Ephemeral Cleanup.} Upon intercepting a \texttt{continue} command, \sysname removes any temporary restorations, reverting process memory to its original state before execution resumes.
\end{enumerate}

\end{document}